%% file: BoBWCSA.tex
\documentclass[12pt]{article}
\usepackage{amsmath}
\usepackage{graphicx,psfrag,epsf}
\usepackage{enumerate}
\usepackage{natbib}
\usepackage{url} 

\input{math_commands.tex}

\RequirePackage[colorlinks,citecolor=blue,linkcolor=blue,urlcolor=blue,pagebackref]{hyperref}
\usepackage{subcaption}

\usepackage{bm}
\usepackage{bbm}
\usepackage{comment}         
\usepackage{multicol}
\usepackage{multirow}
\usepackage{caption}
\usepackage{natbib}
\usepackage{mathtools} 

\usepackage{amsmath,amssymb,amsthm}

\theoremstyle{plain}

\newtheorem{theorem}{Theorem}[section]
\newtheorem{lemma}[theorem]{Lemma}

\newtheorem{corollary}[theorem]{Corollary}

\newtheorem{assumption}[theorem]{Assumption}

\newtheorem*{remark}{Remark}

\usepackage{color}

\usepackage{booktabs}

\def\Holder{{H\"{o}lder}}

\def\Cramer{Cram\'{e}r}

\newcommand{\op}{\mathrm{o}_{p}}

\newcommand{\pa}{\mathrm{\pa}}

\newcommand{\RN}[1]{%
  \textup{\uppercase\expandafter{\romannumeral#1}}%
}

\renewcommand*\eqref[1]{Eq.~(\ref{#1})}

\usepackage{algorithm}
\usepackage{algorithmic}

\def\Cramer{Cram\'{e}r}
\def\Holder{{H\"{o}lder}}

\newcommand*{\email}[1]{\texttt{#1}}

\allowdisplaybreaks

\begin{document}

\def\spacingset#1{\renewcommand{\baselinestretch}%
{#1}\small\normalsize} \spacingset{1}


\title{\bf Double Debiased Covariate Shift Adaptation Robust to Density-Ratio Estimation}
\author{Masahiro Kato\thanks{\email{masahiro-kato@fintec.co.jp}}\\
Data Analytics Division, Mizuho-DL Financial Technology Co., Ltd.\\
and \\
Kota Matsui \\
Department of Biostatistics, Nagoya University Graduate School of Medicine\\
and \\
Ryo Inokuchi \\
Data Analytics Division, Mizuho-DL Financial Technology Co., Ltd.}

\maketitle

\bigskip
 
\begin{abstract}
Consider a scenario where we have access to train data with both covariates and outcomes while test data only contains covariates. In this scenario, our primary aim is to predict the missing outcomes of the test data. With this objective in mind, we train parametric regression models under a covariate shift, where covariate distributions are different between the train and test data. For this problem, existing studies have proposed covariate shift adaptation via importance weighting using the density ratio. This approach averages the train data losses, each weighted by an estimated ratio of the covariate densities between the train and test data, to approximate the test-data risk. Although it allows us to obtain a test-data risk minimizer, its performance heavily relies on the accuracy of the density ratio estimation. Moreover, even if the density ratio can be consistently estimated, the estimation errors of the density ratio also yield bias in the estimators of the regression model's parameters of interest. To mitigate these challenges, we introduce a doubly robust estimator for covariate shift adaptation via importance weighting, which incorporates an additional estimator for the regression function. Leveraging double machine learning techniques, our estimator reduces the bias arising from the density ratio estimation errors. We demonstrate the asymptotic distribution of the regression parameter estimator. Notably, our estimator remains consistent if either the density ratio estimator or the regression function is consistent, showcasing its robustness against potential errors in density ratio estimation. Finally, we confirm the soundness of our proposed method via simulation studies.
\end{abstract}

\noindent%
\vfill

\newpage
\spacingset{1.45} 

\section{Introduction}
Suppose that we have access to train data with both covariates and outcomes, while test data only contains covariates. Under a covariate shift, covariate distributions are different between the train and test data. In this setting, our goal is to predict the missing outcomes of the test data, and our interest lies in training regression models robust against a covariate shift\footnote{Our method can also be applied to binary classification by estimating logistic regression models. See Remark~\ref{rem:logit}.}. 

One of the golden rules in prediction is to train regression models by minimizing an empirical risk that approximates the population risk defined over the target test distribution. However, under the covariate shift, the distribution of the observable training data deviates from that of the test data, making it difficult to approximate the empirical risk of the test data from the training data.

The seriousness of the covariate shift problem depends on the correctness of the model specification. We say regression models are correctly specified if the error terms are independent of covariates under both train and test distributions.
When regression models are correctly specified, we can still obtain a regression model that asymptotically minimizes the test-data population risk by minimizing a train-data empirical risk. This is because the population risk minimizers of the train and test data are the same if the models are correctly specified. In contrast, when regression models are misspecified, their predictive performance can be severely deteriorated \citep{Shimodaira2000}. Thus, this issue has catalyzed interest in regression methods robust to covariate shift and model misspecification across various disciplines, including machine learning and statistics. 

One of the popular approaches for a covariate shift is importance weighting \citep{Shimodaira2000,sugiyama2008direct}. Specifically, \citet{Shimodaira2000} suggests training regression models by weighing the loss with the ratio of covariate densities between the train and test sets. Here, the weighted average loss over the train data converges to the population risk on the test data. Consequently, we can train the regression model by minimizing an estimate of the test-data risk. 

Given that the density ratio is typically unknown, subsequent studies have proposed estimating it from observations \citep{sugiyama2008direct}. By using density-ratio estimation methods, existing importance weighting methods adopt a two-stage approach: initially estimating the density ratio and then minimizing the estimates of the test-data population risk. 

While existing studies report successes of covariate shift adaptation using an estimated density ratio, the performance significantly depends on the accuracy of the density ratio estimation. If the density ratio estimator does not converge to the true density ratio, we cannot obtain regression models that minimize the test-data risk. Even if the density ratio is correctly estimated, the convergence rate of the density ratio estimator affects that of the regression models of interest. If the regression models are parametric models, our interest may lie in the $\sqrt{n}$-convergence rate for the parameters. However, the $\sqrt{n}$-convergence rate may not be obtained because the convergence rate of the density ratio is slower than $\sqrt{n}$. These issues motivate us to develop a covariate shift adaptation method robust to a failure of the density ratio estimation and reduce the bias incurred from the density ratio estimation. 

To tackle these issues, we develop a doubly robust (DR) estimator for covariate shift adaptation, which utilizes both density-ratio and conditional expected outcome estimators. First, this estimator is robust to inconsistency of the density ratio; that is, either the density-ratio and conditional expected outcome estimators are consistent, the DR estimator is also consistent. Additionally, to reduce the bias caused by the estimation error of the density ratio, we apply the double machine learning (DML) method proposed by \citet{ChernozhukovVictor2018Dmlf}. Following this method, we first split observations into subgroups and estimate the density ratio and conditional expected outcomes within them. Then, we add nonparametrically estimated conditional expected outcomes to the conventional empirical test-data risk estimator using the estimated density-ratio function. This procedure of DML enables us to show the asymptotic normality of the DR estimator by achieving a fast convergence rate for estimators of the regression model's parameters. 

\paragraph{Organization.}
The remainder of the paper is organized as follows. In Section~\ref{sec:model}, we define the problem-setting. In Section~\ref{sec:asymptotic_property}, we discuss properties of covariate shift adaptation with importance weighting using an estimated density ratio. As we will discuss in Section~\ref{sec:asymptotic_property}, when using an estimated density ratio, it is not easy to prove the asymptotic normality of the estimator given by covariate shift adaptation with importance weighting. Therefore, in Section~\ref{sec:doublyrobust}, we introduce a new estimator for covariate shift adaptation, which uses a DR structure and DML. We introduce related work in Section~\ref{sec:related} and Appendix~\ref{appdx:related}. In Section~\ref{sec:exp}, we show the performance of the proposed method using real datasets. 

\section{Basic Setup}
\label{sec:model}
We formulate our regression problem. Let $\big(Y, X\big)$ and $\big(\widetilde{Y}, \widetilde{X}\big)$ be two pairs of an outcome and covariates, where $Y, \widetilde{Y} \in\mathbb{R}$, and $X, \widetilde{X}\in\mathcal{X}\subset \mathbb{R}^d (d\geq 1)$. We refer to $\big(Y, X\big)$ as the train data and $\big(\widetilde{Y}, \widetilde{X}\big)$ as the test data. While we can observe $\big(Y, X\big)$ as a pair, we only observe $\widetilde{X}$ for the test data. Our goal is to predict $\widetilde{Y}$ by using realizations of $\big(Y, X\big)$ and $\widetilde{X}$, which is an observable counterpart of $\widetilde{Y}$. We refer to this setting as covariate shift \citep{Shimodaira2000}. The detailed formulations are shown below. 

\subsection{Observations}
In our setting, we observe $\{(Y_i, X_i)\}^{n}_{i=1}$ and $\{\widetilde{X}_j\}^{m}_{j=1}$, which are independently and identically distributed (i.i.d.) copies of $\big(Y, X\big)$ and $\widetilde{X}$. They are two stratified observations; that is, for fixed sample sizes $n$ and $m$, the two datasets are independently generated \citep{Imbens1996,Wooldridge2001}.\footnote{As a different setting, we can consider a situation where there are $n_0$ samples, and following a Bernoulli distribution, samples are classified into two datasets with $n$ and $m$ sample sizes. In this setting, $n$ and $m$ are random variables. We do not consider this setting in our study.} Following the terminology of machine learning, we call $\{(Y_i, X_i)\}^{n}_{i=1}$ the (labeled) train data and $\{\widetilde{X}_j\}^{m}_{j=1}$ the (unlabeled) test data.

\subsection{Covariate Shift}
For each pair of potential outcomes $\big(Y, X\big)$ and $\big(\widetilde{Y}, \widetilde{X}\big)$, we assume the data-generating process as
\begin{align*}
&(Y, X) \overset{\mathrm{i.i.d.}}{\sim} p(y, x),\ \ \ (\widetilde{Y}, \widetilde{X}) \overset{\mathrm{i.i.d.}}{\sim} q(y, x), 
\end{align*}
where $p(y, x)$ is the joint probability density function of $(Y, X)$ and $q(y, x)$ is the joint joint probability density function of $\big(\widetilde{Y}, \widetilde{X}\big)$. 

For $\big(Y, X\big)$ and $\big(\widetilde{Y}, \widetilde{X}\big)$, we make the following boundedness assumption.
\begin{assumption}[Boundedness of random variables]
    \label{asm:bounded_output}
    There exists a universal constant $C > 0$ such that $|Y|, |X_{(k)}|, |\widetilde{Y}|, |\widetilde{X}_{(k)}| < C$ for all $k\in\{1,2,\dots, d\}$, where $X_{(k)}$ and $\widetilde{X}_{(k)}$ are $k$-th elements of $X$ and $\widetilde{X}$. 
\end{assumption}

We also assume that the probability density functions of $Y$ conditioned on $X$ and $\widetilde{Y}$ conditioned on $\widetilde{X}$ are the same; that is,
\begin{align*}
p(y, x) = w(y| x)p(x),\ \ \ q(y, x) = w(y| x)q(x),
\end{align*}
where $w(y| x)$ is the conditional probability density function of $Y$ (resp. $\widetilde{Y}$) given $X$ (resp. $\widetilde{X}$), $p(x)$ and $q(x)$ are probability density functions of $X$ and $\widetilde{X}$, respectively. Here, let us define the density ratio $r(x) := \frac{q(x)}{p(x)}$, and we assume the overlap between the supports of $X$ and $\widetilde{X}$. 
\begin{assumption}
\label{asm:bounded_dens_ratio}
There exists an universal constant $0 < C_{r} < \infty$ such that for any $x\in\mathcal{X}$, 
$
0 \leq r(x) = \frac{p(x)}{q(x)} < C_r
$. 
\end{assumption}

Furthermore, we denote the conditional expectation function $\mathbb{E}[Y|X = x]$ as
\begin{align*}
    f_0(x) := \mathbb{E}[Y|X = x],
\end{align*}
where $\mathbb{E}[Y|X = x]$ takes an expectation with respect to the density $w(y|x)$. 
Here, the following hold:
\begin{align*}
    Y = f_0(X) + \varepsilon,\qquad \widetilde{Y} = f_0(\widetilde{X}) + \widetilde{\varepsilon},
\end{align*}
where $\varepsilon$ and $\widetilde{\varepsilon}$ are error terms such that $\mathbb{E}[\varepsilon | X ] = 0$ and $\mathbb{E}[\widetilde{\varepsilon} | \widetilde{X}]$ almost surely.

\subsection{Parametric Regression Models}
To investigate the relationship between $Y_i$ and $X_i$, we model the conditional expectation $f_0(x)$ by a parametrized regression model $g_{\bm{\beta}}$ with a $k$-dimensional parameter $\bm{\beta}$, and define a set of functions as $\mathcal{G}=\{g_{\bm{\beta}}(x): \bm{\beta}\in\Theta\subset \mathbb{R}^k\}$, where $\Theta$ is some compact parameter space. Then, we regress $X$ on $Y$ to determine a parameter $\bm{\beta}$. For the train and test data, let us define the mean squared error minimizers as
\begin{align*}
    \bm{\beta}_0 := \argmin_{\bm{\beta}\in \Theta} \mathbb{E}_p\left[\left( Y - g_{\bm{\beta}}(X) \right)^2\right], \qquad
    \widetilde{\bm{\beta}}_0 := \argmin_{\bm{\beta}\in \Theta} \mathbb{E}_q\left[\left( \widetilde{Y} - g_{\bm{\beta}}(\widetilde{X}) \right)^2\right],
\end{align*}
where $\mathbb{E}_p$ and $\mathbb{E}_q$ denote expectations over $p(y, x)$ and $q(y, x)$. 
To identify the parameters in estimation, we assume uniqueness of $\widetilde{\bm{\beta}}_0$ and differentiability of $g_{\bm{\beta}}$ with respect to $\bm{\beta}$. 
\begin{assumption}
\label{eq:unique_beta}
The risk minimizer $\widetilde{\bm{\beta}}_0$ uniquely exists.
    Regression models $g_{\bm{\beta}}$ are second-order differentiable with respect to $\bm{\beta}$.  
\end{assumption}
Note that when the optimal parameter $\beta_0$ is unique, the corresponding asymptotic distribution is also unique.

\begin{remark}
    \label{rem:uniqueness}
    If we allow non-unique asymptotic distributions, we do not have to assume the global uniqueness of $\widetilde{\bm{\beta}}_0$. Instead, we can assume uniqueness only around a local optimum $\widehat{\bm{\beta}_0}$.
\end{remark}

\begin{remark}
\label{rem:logit}
    A parametric model $g_{\bm{\beta}}$ includes a logistic regression model $g_{\bm{\beta}}(x) = {1}/\big\{1 + \exp\left( - x^\top \bm{\beta}\right)\big\}$ \citep{VaartA.W.vander1998As}.
\end{remark}

\subsection{Model Specification}
We consider two scenarios for the regression problem. In the first scenario, a regression model is correctly specified; that is, $\mathcal{G}$ includes the regression function $f_0(x) = \mathbb{E}[Y|X=x]$. In the second scenario, $\mathcal{G}$ does not include the regression function $f_0(x) = \mathbb{E}[Y|X=x]$. We say that regression models are correctly specified in the former scenario, and regression models are misspecified in the latter scenario. 

\paragraph{Models are correctly specified.}
If $f_0 \in \mathcal{G}$, then it holds that $\bm{\beta}_0 = \widetilde{\bm{\beta}}_0$, $f_0 = g_{\bm{\beta}_0} = g_{\widetilde{\bm{\beta}}_0}$, and 
\begin{align*}
Y = g_{\bm{\beta}_0}(X) + \varepsilon,\qquad \widetilde{Y} = g_{\bm{\beta}_0}(\widetilde{X}) + \widetilde{\varepsilon},
\end{align*}
In this scenario, we can train a model that asymptotically minimizes the test-data risk by minimizing an empirical risk over the train data; that is, without covariate shift adaptation, we can consistently estimate the risk minimizer for the test-data risk.

\paragraph{Models are misspecified.} If $f_0 \notin \mathcal{G}$, then it may not hold that $\bm{\beta}_0 = \widetilde{\bm{\beta}}^*$ and $f_0 = g_{\bm{\beta}_0} = g_{\widetilde{\bm{\beta}}^*}$. If $g_{\bm{\beta}_0} \neq g_{\widetilde{\bm{\beta}}^*}$, then it implies that $g_{\bm{\beta}_0}$ may not perform well for data generated from $q(y, x)$. In this scenario, it is desirable to approximate the test-data risk and train regression models by minimizing it. 

Thus, the problem of a covariate shift is serious when models are misspecified. To address this issue, \citet{Shimodaira2000} proposes the following covariate shift adaptation via importance weighting. 

\subsection{Covariate Shift Adaptation}
Our goal lies in predicting a test response variable $\widetilde{Y}_j$, an unobserved counterpart of $\widetilde{X}_j$, based on the parametric regression models estimated from $\{(Y_i, X_i)\}^{n}_{i=1}$ and $\{\widetilde{X}_j\}^{m}_{j=1}$. The problem of predicting a test response variable $\widetilde{Y}_j$ under covariate shift ($p(x) \neq q(x)$) and model misspecification is referred to as \emph{covariate shift adaptation}. 

As we review in the following section, a typical approach for this problem is to employ importance weighting \citep{Shimodaira2000}. However, while this approach allows us to adapt a covariate shift, it is vulnerable to a failure of density ratio estimation, and its estimator has a slow convergence rate owing to the estimation error of the density ratio.

In this study, we propose a method for covariate shift adaptation that 
is robust to a failure of density ratio estimation and yields an estimator with a faster convergence rate. We refer to our approach as \emph{DR covariate shift adaptation} because it returns a consistent estimator if one of the two functions in the estimator is correctly specified, as explained in Section~\ref{sec:dr}.

\section{Covariate Shift Adaptation via Importance Weighting and Density-Ratio Estimation}
\label{sec:asymptotic_property}
For obtaining the risk minimizer under the covariate shift, importance weighting with the density ratio has been used \citep{Shimodaira2000,Reddi2015}. 

\subsection{Covariate Shift Adaptation via Importance Weighting}
To estimate the parameter $\widetilde{\bm{\beta}}_0$, which is the minimizer of the prediction risk under the overate shift, we face the problem that we do not have access to the labeled data and cannot obtain a sample approximation of the test-data population risk $\mathcal{R}(\bm{\beta}) := \mathbb{E}_{q}\left[(Y-g_{\bm{\beta}}\left(X\right))^2\right]$.

Let us define  the \emph{density ratio} between $q$ and $p$ as
\[r_0(x) = \frac{q(x)}{p(x)}\]
\cite{Shimodaira2000} shows that we can obtain the population risk for the test distribution using the following weighted squared risk for the training distribution:
\begin{align*}
\mathbb{E}_{q}\left[(Y-g_{\bm{\beta}}\left(X\right))^2\right] = \mathbb{E}_{p}\left[(Y-g_{\bm{\beta}}\left(X\right))^2r_0(X)\right].
\end{align*}

Then, the risk minimizer under covariate shift adaptation with importance weighting is given as 
\begin{align*}
\widehat{\bm{\beta}}^* = \argmin_{\bm{\beta} \in \Theta} \frac{1}{n}\sum^n_{i=1} \left(Y_i -g_{\bm{\beta}}\left(X_i\right)\right)^2 r_0(X_i),
\end{align*}

\subsection{Density Ratio Estimation}
Since we do not know the true density ratio $r_0$, we need to estimate it. A naive approach is to estimate the two densities of the numerator $q(x)$ and denominator $p(x)$ separately, then take the ratio of the estimated densities. For instance, we can separately estimate the two probability densities by kernel density estimation. However, this approach may not work well due to the division between the two estimated densities, which magnifies the estimation errors.


For example, this study employs \emph{Least-Squares Importance Fitting} (LSIF) in experimental studies, which uses the squared error to estimate the density ratio \citep{Kanamori2009}. The reason for this choice is that there is an algorithm called \emph{unconstrained Least-Squares Importance Fitting} (uLSIF) with a computational advantage. We can obtain the closed-form solution just by solving the linear equations. 
By using the uLSIF, given a hypothesis class $\mathcal{H}$, we obtain $\widehat{r}_{n, m}$, an estimator of $r_0$, by
\begin{align*}
&\widehat{r}_{n,m} = \argmin_{s\in\mathcal{H}} \Bigg\{\frac{1}{2n} \sum^n_{i=1}s^2(X_i)- \frac{1}{m} \sum^m_{j=1}s(\widetilde{X}_j) + \mathcal{R}(s)\Bigg\},
\end{align*}
where $\mathcal{R}:\mathcal{H}\to \mathbb{R}^+$ is a regularization term. Then, we define a regression parameter estimator as
\begin{align*}
\widehat{\bm{\beta}} = \argmin_{\bm{\beta} \in \Theta} \frac{1}{n}\sum^n_{i=1} \left(Y_i -g_{\bm{\beta}}\left(X_i\right)\right)^2 \widehat{r}_{n,m}(X_i),
\end{align*}

\subsection{Covariate Shift Adaptation with Estimated Density-Ratio Function}
The convergence rates of density ratio estimators are shown by existing studies, such as \cite{Sugiyama2012}, \citet{Kanamori2012}, \citet{Kato2020nnbr}, and \citet{zheng2022an}. For example, when using kernel regression, the convergence rate is given as follows \citep{Kanamori2012}: 
\begin{align*}
\left\|\widehat{r}_{n.m}(X) - r_0(X)\right\|_2 = O_p\left(\min\left\{n, m\right\}^{-\frac{1}{2(1+\gamma)}}\right),
\end{align*}
where $0 < \gamma < 1$ is a constant depending on the functional form of the density ratio. These results imply that the estimation error remains in the estimation of $\mathbb{E}_{q}\left[(Y-g_{\bm{\beta}}\left(X\right))^2\right] = \mathbb{E}_{p}\left[(Y-g_{\bm{\beta}}\left(X\right))^2r_0(X)\right]$ as
\begin{align*}
    \mathbb{E}_{p}\left[(Y-g_{\bm{\beta}}\left(X\right))^2\widehat{r}_{n.m}(X)\right] 
    = \mathbb{E}_{p}\left[(Y-g_{\bm{\beta}}\left(X\right))^2r_0(X)\right] + O_p\left(\min\left\{n, m\right\}^{-\frac{1}{2(1+\gamma)}}\right).\nonumber
\end{align*}
We typically consider $\sqrt{n}$-consistency for estimators of $\bm{\beta}_0$. However, this result implies that in the estimation of $\mathbb{E}_{q}\left[(Y-g_{\bm{\beta}}\left(X\right))^2\right]$, there remains non-negligible bias term that does not disappear, which remains in the estimation of $\bm{\beta}_0$. 

In the following sections, we propose our estimator that attains (i) covariate shift robustness under misspecified models and (ii) asymptotic normality, which implies $\sqrt{n}$-consistency.

\section{DR Covariate Shift Adaptation}
\label{sec:doublyrobust}
When estimating the density ratio with modern machine learning methods, the estimator often violates the Donsker condition. To show the asymptotic normality of an estimator of $\bm{\beta}_0$ with such a nuisance estimator, we introduce DML \citep{ChernozhukovVictor2018Dmlf}, which employs a sample-splitting technique called cross-fitting \citep{klaassen1987,ZhengWenjing2011CTME}.

For brevity, without loss of generality, let $n$ and $m$ be even numbers, and we use two-fold cross-fitting. We can extend our method to a case with $L$-fold cross-fitting. For $n_{(1)} = n_{(2)} = n/2$ and $m_{(1)} = m_{(2)} = m/2$, we randomly separate the samples $\{(Y_i, X_i)\}^{n}_{i=1}$ into two subsets $D_{(1)} =  \{(Y_{(i, 1)}, X_{(i, 1)})\}^{n_{(1)}}_{i=1}$ and $D_{(2)} = \{(Y_{(i, 2)}, X_{(i, 2)})\}^{n_{(2)}}_{i=1}$ with sample sizes $n_{(1)}$ and $n_{(2)}$; similarly, we separate samples $\{\widetilde{X}_j\}^{m}_{j=1}$ into two subsets $\widetilde{D}_{(1)} = \{\widetilde{X}_{(j, 1)}\}^{m_{(1)}}_{j=1}$ and $\widetilde{D}_{(2)} = \{\widetilde{X}_{(j, 2)}\}^{m_{(2)}}_{j=1}$ with sample sizes $m_{(1)}$ and $m_{(2)}$.

First, by using $\{(Y_{(i, 2)}, X_{(i, 2)})\}^{n_{(2)}}_{i=1}$ and $\{\widetilde{X}_{(j, 2)}\}^{m_{(2)}}_{j=1}$, we estimate $f_0$ and $r_0$, respectively. We estimate them to satisfy Assumption~\ref{asm:nuisance_conv}. Let $\widehat{f}_{(-1)}$ and $\widehat{r}_{(-1)}$ be the estimators.

Then, by using $\{(Y_{(i, 1)}, X_{(i, 1)})\}^{n_{(1)}}_{i=1}$ and $\{\widetilde{X}_{(j, 1)}\}^{m_{(1)}}_{j=1}$, we obtain a sample approximate of the risk as
\begin{align*}
&\widehat{\mathcal{R}}_{(1)}(\bm{\beta}) :=\frac{1}{n_{(1)}}\sum^{n_{(1)}}_{i=1}\Big\{\left(Y_{(i, 1)} - g_{\bm{\beta}}\big(X_{(i, 1)}\big) \right)^2 
- \left(\widehat{f}_{(-1)}\big(X_{(i, 1)}\big) - g_{\bm{\beta}}\big(X_{(i, 1)}\big)\right)^2\Big\} \widehat{r}_{(-1)}\big(X_{(i, 1)}\big)\\
&~~~~~~~~~~~~~~~+\frac{1}{m_{(1)}}\sum^{m_{(1)}}_{j=1}\left\{\left(\widehat{f}_{(-1)}(\widetilde{X}_{(j, 1)}) - g_{\bm{\beta}}(\widetilde{X}_{(j, 1)})\right)^2\right\}.
\end{align*}

Similarly, we define another empirical risk as 
\begin{align*}
&\widehat{\mathcal{R}}_{(2)}(\bm{\beta}) := \frac{1}{n_{(2)}}\sum^{n_{(2)}}_{i=1}\Big\{\left(Y_{(i, 2)} - g_{\bm{\beta}}(X_{(i, 2)}) \right)^2
- \left(\widehat{f}_{(-2)}(X_{(i, 2)}) - g_{\bm{\beta}}(X_{(i, 2)})\right)^2\Big\} \widehat{r}_{(-2)}(X_{(i, 2)})\\
&~~~~~~~~~~~~~~~+\frac{1}{m_{(2)}}\sum^{m_{(2)}}_{j=1}\left\{\left(\widehat{f}_{(-2)}(\widetilde{X}_{(j, 2)}) - g_{\bm{\beta}}(\widetilde{X}_{(j, 2)})\right)^2\right\}.
\end{align*}
We show the pseudo-code in Appendix~\ref{appdx:pseudo}.

By minimizing the two empirical risks, we obtain the DR estimator $\widehat{\bm{\beta}}^{\mathrm{DR}}$ for $\bm{\beta}_0$; that is,
\begin{align}
\label{eq:opt}
    \widehat{\bm{\beta}}^{\mathrm{DR}} = \argmin_{\bm{\beta} \in \Theta}\Big\{\widehat{\mathcal{R}}_{(1)}(\bm{\beta}) + \widehat{\mathcal{R}}_{(2)}(\bm{\beta})\Big\}
\end{align}
We can also estimate $f_0(x)$ by $\widehat{f}^{\mathrm{DR}}(x) = g_{\widehat{\bm{\beta}}^{\mathrm{DR}}}(x)$. 

\begin{remark}[Linear models]
Note that we have more simplified form of $\widehat{\bm{\beta}}^{\mathrm{DR}}$ when the regression model $g_{\bm{\beta}}$ is a linear model; that is, $g_{\bm{\beta}}(x) = Z^\top(x)\bm{\beta}$, where $Z:\mathcal{X} \to \mathbb{R}^K$ is some basis function, which maps $x\in\mathcal{X}$ to a $k$-dimensional vector. In this case, we have
\begin{align*}
    &\widehat{\bm{\beta}}^{\mathrm{DR}} = \left\{\sum_{\ell\in\{1, 2\}}\frac{1}{m_{(\ell)}}\sum^{m_{(\ell)}}_{j=1}Z\left(\tilde{X}_{(j, \ell)}\right)Z^\top\left(\tilde{X}_{(j, \ell)}\right)\right\}^{-1}\\
&~~~~~~~~~~\times \sum_{\ell\in\{1, 2\}}\Bigg\{\frac{1}{n_{(\ell)}}\sum^{n_{(\ell)}}_{i=1}Z(X_{(i, \ell)})\left(Y_i - \hat{f}_{(-\ell))}(X_{(i, \ell)})\right)\hat{r}_{(-\ell))}(X_{(i, \ell)})\\
&~~~~~~~~~~~~~~~ + \frac{1}{m_{(\ell)}}\sum^{m_{(\ell)}}_{j=1}Z\left(\tilde{X}_{(j, \ell)}\right)\hat{f}_{(-\ell)}\left(\tilde{X}_{(j, \ell)}\right)\Bigg\}.
\end{align*}
\end{remark}

\subsection{Asymptotic Normality}
\label{subsec:asymp_prp}
Let us assume $m=\rho n$ for a fixed $\rho > 0$ and consider asymptotics when $n \to \infty$.

First, we consider the first order condition of \eqref{eq:opt}, which is given as $h(\widehat{\bm{\beta}}^{\mathrm{DR}}) = \bm{0}$, where $h(\widehat{\bm{\beta}}^{\mathrm{DR}})$ is a $k$-dimensional vector defined as
\begin{align*}
h(\bm{\beta}) &:= \frac{\partial \sum_{\ell \in\{1, 2\}}\widehat{\mathcal{R}}_{(\ell)}(\bm{\beta})}{\partial \bm{\beta}}\\
&=\sum_{\ell \in\{1, 2\}}\Bigg\{\frac{1}{n_{(\ell)}}\sum^{n_{(\ell)}}_{i=1}\dot{g}_{\bm{\beta}}(X_i)\Big\{Y_{(i, \ell)}  - \widehat{f}_{(-\ell)}(X_{(i, \ell)})\Big\}\widehat{r}_{(-\ell)}(X_{(i, \ell)})\\
&~~~~~~~~~~+ \frac{1}{m_{(\ell)}}\sum^{m_{(\ell)}}_{j=1}\dot{g}_{\bm{\beta}}(\widetilde{X}_{(j, \ell)})\left\{\widehat{f}_{(-\ell)}(\widetilde{X}_{(j, \ell)}) - g_{\bm{\beta}}(\widetilde{X}_{(j, \ell)})\right\}\Bigg\},
\end{align*}
where $\dot{g}_{\bm{\beta}}(X_i) = \frac{\partial g_{\bm{\beta}}(X_i)}{\partial \bm{\beta}}$.
We also define 
\begin{align*}
    \dot{h}(\bm{\beta}) 
    &:= \frac{\partial h(\bm{\beta})}{\partial \bm{\beta}^\top}\\
    &=\sum_{\ell\in\{1, 2\}}\Bigg\{ \frac{1}{n_{(\ell)}}\sum^{n_{(\ell)}}_{i=1}\left(Y_{(i, \ell)} - \widehat{f}_{(-\ell)}(X_{(i, \ell)})\right)\widehat{r}_{(-\ell)}(X_{(i, \ell)})\ddot{g}_{\bm{\beta}}(X_{(i, \ell)})\\
    &\ \ \ - \frac{1}{m_{(\ell)}}\sum^{m_{(\ell)}}_{j=1}\Big\{ \dot{g}_{\bm{\beta}}(\widetilde{X}_{(j, \ell)})\dot{g}^\top_{\bm{\beta}}(\widetilde{X}_{(j, \ell)}) + \left(\widehat{f}_{(-\ell)}(\widetilde{X}_{(j, \ell)}) - g_{\bm{\beta}}(\widetilde{X}_{(j, \ell)})\right)\ddot{g}_{\bm{\beta}}(\widetilde{X}_{(j, \ell)}) \Big\} \Bigg\}.\\
\end{align*}
Then, we make the following assumption.
\begin{assumption}
\label{asm:linear_dep}
    For any sequence $\{\bm{\beta}^*_i\}^\infty_{i=1}$ such that $\bm{\beta}^*_i \xrightarrow{\mathrm{p}} \widetilde{\bm{\beta}}_0$, 
    \begin{align*}
        &\left(\dot{h}(\bm{\beta})\big|_{\bm{\beta} \bm{\beta}^*_i}\right)_{i,j} = \left(\dot{h}(\dot{h}(\bm{\beta})\big|_{\bm{\beta} = \bm{\beta}_0}\right)_{i,j} = D^\top_{i,j} + o_p(1).
    \end{align*}
    holds for each $i,j \in \{1,2,\dots,k\}$, where $D$ is a linearly independent ($k\times k$)-matrix. 
\end{assumption}

We make the following assumptions for the convergence rates of estimators of nuisance parameters.
\begin{assumption}
\label{asm:nuisance_conv}
    For $\ell \in \{1,2\}$, the followings hold:
\begin{align*}
    &\left\|\widehat{r}_{(-\ell)} - r_0\right\|_2 = o_p(1),\qquad \left\|\widehat{f}_{(-\ell)} - f_0\right\|_2 = o_p(1)\\
    &\left\|\widehat{r}_{(-\ell))} - r_0\right\|_2\left\|\widehat{f}_{(-\ell)} - f_0\right\|_2 = o_p\left(1/\min\{n, m\}^{1/2}\right)),
\end{align*}
as $n, m \to \infty$.
\end{assumption}
For example, the uLSIF with kernel ridge regression satisfies this convergence rate \citet{Kanamori2012}. Density-ratio estimators with neural networks under Bregman-divergence minimization also satisfy this convergence rate \citep{Kato2020nnbr}.

Note that under Assumptions~\ref{asm:linear_dep} and \ref{asm:nuisance_conv}, $D^\top = \mathbb{E}_q\left[ \dot{g}_{\bm{\beta}}(\widetilde{X})\dot{g}^\top_{\bm{\beta}}(\widetilde{X})\right]$ holds. 

In DML, the combination of Assumption~\ref{asm:nuisance_conv} and the cross-fitting allows us to reduce the bias incurred from the estimation error of $r_0$. This is because we can make the convergence rate of $\left\|\widehat{r}_{(-\ell))} - r_0\right\|_2\left\|\widehat{f}_{(-\ell)} - f_0\right\|_2$ faster than $\left\|\widehat{r}_{(-\ell))} - r_0\right\|_2$ due to $\left\|\widehat{f}_{(-\ell)} - f_0\right\|_2$. 

Then, we show the asymptotic distributions as follows. The proof is shown in Appendix~\ref{appendixthm:asymp_dist}. 
\begin{theorem}
\label{thm:main}
Suppose that Assumptions~\ref{asm:bounded_output}--\ref{eq:unique_beta}, \ref{asm:linear_dep}--\ref{asm:nuisance_conv} hold. If $m = \rho n$ for some universal constant $\rho$ and $f_0 \in \mathcal{G}$, then the DR estimator $\widehat{\bm{\beta}}^{\mathrm{SDB}}$ has the following asymptotic distribution:
\label{thm:asymp_dist}
    \begin{align*}
        &\sqrt{n}\left(\bm{\beta}_0 - \widehat{\bm{\beta}}^{\mathrm{DR}}\right)\xrightarrow{\mathrm{d}} \mathcal{N}\left(\bm{0}, \Omega\right),
    \end{align*}
    where
    \begin{align*}
        \Omega = \mathbb{E}_q\left[\dot{g}_{\bm{\beta}_0}(\widetilde{X})\dot{g}^\top_{\bm{\beta}_0}(\widetilde{X})\right]^{-1} \mathbb{E}_q\left[\sigma^2(\widetilde{X})r(\widetilde{X})\dot{g}_{\bm{\beta}_0}(\widetilde{X})\dot{g}^\top_{\bm{\beta}_0}(\widetilde{X})\right]\mathbb{E}_q\left[\dot{g}_{\bm{\beta}_0}(\widetilde{X})\dot{g}^\top_{\bm{\beta}_0}(\widetilde{X})\right]^{-1}.
    \end{align*}
\end{theorem}

From this theorem and the Taylor expansion, we can obtain the following point-wise asymptotic normality for $\widehat{f}^{\mathrm{DR}}(x) = g_{\widehat{\bm{\beta}}^{\mathrm{DR}}}(x)$. 
\begin{corollary}
\label{cor:asymp_dist}
Under the same conditions in Theorem~\ref{thm:main}, then the DR estimator $\widehat{f}^{\mathrm{DR}}(x)$ has the following asymptotic distribution for each $x\in\mathcal{X}$:
    \begin{align*}
        &\sqrt{n}\left(f_0(x) - \widehat{f}^{\mathrm{DR}}(x)\right)\xrightarrow{\mathrm{d}} \mathcal{N}\left(\bm{0}, \dot{g}^\top_{\bm{\beta}_0}(x)\Omega\dot{g}_{\bm{\beta}_0}(x)\right).
    \end{align*}
\end{corollary}
\begin{proof}
    From the Taylor expansion of $g_{\widehat{\bm{\beta}}^{\mathrm{DR}}}(x) $ around $\widehat{\bm{\beta}}^{\mathrm{DR}} = \bm{\beta}_0$, we have $
        \sqrt{n}\left(g_{\bm{\beta}_0}(x) - g_{\widehat{\bm{\beta}}^{\mathrm{DR}}}(x)\right)= \sqrt{n}\dot{g}^\top_{\bm{\beta}_0}(x)\left(\widehat{\bm{\beta}}^{\mathrm{DR}} - \bm{\beta}_0\right) + o\left(\sqrt{n}\left(\bm{\beta}_0 - \widehat{\bm{\beta}}^{\mathrm{DR}} \right)^2\right)$. 
    From $g_{\widehat{\bm{\beta}}^{\mathrm{DR}}}(x) = \widehat{f}^{\mathrm{DR}}(x)$, $g_{\bm{\beta}_0}(x) = f_0(x)$, and $\sqrt{n}\left(\bm{\beta}_0 - \widehat{\bm{\beta}}^{\mathrm{DR}}\right) = O_p(1)$, the statement holds. 
\end{proof}

\subsection{Double Robustness}
\label{sec:dr}
Our proposed covariate shift adaptation is DR in the sense that if either the density ratio or the conditional expectation is consistently estimated, the parameter $\widetilde{\bm{\beta}}_0$ is also consistently estimated. 

Let us define another risk estimator as follows:
\begin{align*}
\ddot{\mathcal{R}}(\bm{\beta}) &:= \sum_{\ell\in\{1, 2\}}\frac{1}{n_{(\ell)}}\sum^{n_{(\ell)}}_{i=1}\Bigg\{\left(Y_{(i, \ell)} - g_{\bm{\beta}}(X_{(i, \ell)}) \right)^2 - \left(f^\dagger(X_{(i, \ell)}) - g_{\bm{\beta}}(X_{(i, \ell)})\right)^2\Bigg\} r^\dagger(X_{(i, \ell)})\\
&\ \ \ \ \ \ \ \ \ \ \ \ + \sum_{\ell\in\{1, 2\}}\frac{1}{m_{(\ell)}}\sum^{m_{(\ell)}}_{j=1}\left(f^\dagger(\widetilde{X}_{(j, \ell)}) - g_{\bm{\beta}}(\widetilde{X}_{(j, \ell)})\right)^2\\
&= \frac{1}{n}\sum^{n}_{i=1}\big\{\left(Y_i - g_{\bm{\beta}}(X_i) \right)^2 - \left(f^\dagger(X_i) - g_{\bm{\beta}}(X_i)\right)^2\big\} r^\dagger(X_i) + \frac{1}{m}\sum^{m}_{j=1}\left(f^\dagger(\widetilde{X}_i) - g_{\bm{\beta}}(\widetilde{X}_i)\right)^2,
\end{align*}
which replaces $\widehat{f}_{(-\ell)}$ and $\widehat{r}_{(-\ell)}$ with $f^\dagger$ and $r^\dagger$, respectively. Then, we have
\begin{align*}
    \ddot{\mathcal{R}}(\bm{\beta}) \approx
    \mathbb{E}_p\left[\Big\{\left(Y - g_{\bm{\beta}}(X) \right)^2 - \left(f^\dagger(X) - g_{\bm{\beta}}(X)\right)^2\Big\} r^\dagger(X)\right] 
    + \mathbb{E}_q\left[\left(f^\dagger(\widetilde{X}) - g_{\bm{\beta}}(\widetilde{X})\right)^2\right].
\end{align*}
Using this equation, we confirm that if either $f_0$ or $r_0$ is consistently estimated, we can consistently estimate the test-data population risk $\mathcal{R}(\bm{\beta})$ as shown below. Then, we can consistently estimate $\widetilde{\bm{\beta}}_0$ by minimizing the estimate of $\mathcal{R}(\bm{\beta})$. 

\paragraph{Case~(i): $f^\dagger = f_0 = g_{\bm{\beta}_0}$.}
If $f^\dagger = f_0 = g_{\bm{\beta}_0}$, then
$\bm{\beta}_0 - \widehat{\bm{\beta}}^{\dagger} \approx 0$ holds, where we used 
\begin{align*}
    &\ddot{\mathcal{R}}(\bm{\beta}) \approx \mathbb{E}_q\left[\left(f_0(\widetilde{X}) - g_{\bm{\beta}}(\widetilde{X})\right)^2\right] + C,
\end{align*}
where $C > 0$ is a constant that is irrelevant to $\bm{\beta}$. Here, 
\begin{align*}
    \widetilde{\bm{\beta}}_0 = \argmin_{\bm{\beta}\in \Theta}\ddot{\mathcal{R}}(\bm{\beta}) = \argmin_{\bm{\beta}\in \Theta} \mathbb{E}_q\left[\left( \widetilde{Y} - g_{\bm{\beta}}(\widetilde{X}) \right)^2\right]
\end{align*}
holds. Note that $r^\dagger$ is some function that may not be $r_0$. Therefore, we can adapt a covariate shift if we consistently estimate $f_0$ and cannot estimate $r_0$ well. 

\paragraph{Case~(ii): $r^\dagger = r_0$.}
Additionally, if $r^\dagger = r_0$, then 
\begin{align*}
        \ddot{\mathcal{R}}(\bm{\beta}) \approx \mathbb{E}_q\left[\left(\widetilde{Y} - g_{\bm{\beta}}(\widetilde{X})\right)^2\right],
    \end{align*}
holds, where we used the following two equations:
\begin{align*}
&\mathbb{E}_p\left[\left(Y - g_{\bm{\beta}}(X) \right)^2 r_0(X)\right] = \mathbb{E}_q\left[\left(Y - g_{\bm{\beta}}(X) \right)^2\right], \\ 
&\mathbb{E}_p\left[\left\{\left(f^\dagger(X) - g_{\bm{\beta}}(X)\right)^2\right\} r_0(X)\right]
= \mathbb{E}_q\left[\left(f^\dagger(\widetilde{X}) - g_{\bm{\beta}}(\widetilde{X})\right)^2\right].
\end{align*}
Thus, if either the density ratio or the conditional expectation is consistently estimated, we can consistently estimate the test-data population risk.

{\footnotesize 
\begin{table*}[t]
    \centering
    \caption{Results in Section~\ref{sec:exp}. The mean squared errors are reported with the standard deviations in parenthesis.}
    \label{tbl:exp_res}
    \vspace{-0.2cm}
    \scalebox{0.7}[0.7]{
    \begin{tabular}{|l|l|r|r|r|r|r|r|r|}
    \hline
       & &  \multicolumn{2}{|c|}{OLS}  &     \multicolumn{2}{|c|}{WLS}  &     \multicolumn{1}{|c|}{NP}  &       \multicolumn{2}{|c|}{DR} \\
          \hline
   & & Misspecified & Correct & Misspecified & Correct & - & Misspecified & Correct \\
    \hline
    \multirow{2}{*}{Model~1} & Indep $X$ &  7.295 (6.794) &  1.008 (0.064) &  3.820 (2.161) &  1.022 (0.070) &  1.025 (0.076) &  3.085 (1.337) &  1.078 (0.186) \\
    & Corr $X$ &  6.359 (5.624) &  1.019 (0.074) &  3.596 (1.825) &  1.026 (0.074) &  1.03 (0.080) &  3.054 (1.332) &  1.049 (0.126) \\
    \multirow{2}{*}{Model~2} & Indep $X$ &  0.194 (0.042) &  0.154 (0.053) &  0.195 (0.050) &  0.156 (0.063) &  0.158 (0.038) &  0.189 (0.042) &  0.198 (0.140) \\
    & Corr $X$ &  0.206 (0.035) &  0.168 (0.036) &  0.214 (0.059) &  0.170 (0.036) &  0.183 (0.043) &  0.203 (0.035) &  0.175 (0.040)  \\
    \bottomrule
    \end{tabular}
    }
    \vspace{-0.3cm}
\end{table*}
}

\subsection{Self-Debiased Covariate Shift Adaptation}
\label{sec:sdb_csa}
In the previous subsection, to estimate $f_0$ by $g_{\bm{\beta}}$, we assume that another estimator of $f_0$ is obtainable with a certain convergence rate. However, covariate shift adaptation mainly considers a situation where we cannot train high-dimensional, nonparametric, or other complicated models, including neural networks and random forests. This is because covariate shift adaptation is a method for addressing model misspecification, and if we use such models with large $\mathcal{G}$, we may not encounter the problem of model misspecification. Therefore, in this subsection, for a situation where we cannot use another model for estimating $f_0$, we propose another estimator with asymptotic normality only using a model of interest $g_{\bm{\beta}}$ and a density ratio model as a variant of the previously defined DR estimator.

We define the following self-debiased (SDB)  estimator:
\[\widehat{\bm{\beta}}^{\mathrm{SDB}} := \argmin_{\bm{\beta} \in \Theta}\sum_{\ell\in\{1, 2\}}\widehat{\mathcal{R}}^{\mathrm{SDB}}_{(\ell)}(\bm{\beta}),\] 
where
\begin{align*}
\widehat{\mathcal{R}}^{\mathrm{SDB}}_{(\ell)}(\bm{\beta}) &:= \frac{1}{n_{(\ell)}}\sum^{n_{(\ell)}}_{i=1}\Big\{\left(Y_{(i, \ell)} - g_{\bm{\beta}}(X_{(i, \ell)}) \right)^2
 - \left(g_{\widehat{\bm{\beta}}_{(-\ell)}}(X_{(i, \ell)}) - g_{\bm{\beta}}(X_{(i, \ell)})\right)^2\Big\} \widehat{r}^\alpha_{(-\ell)}(X_{(i, \ell)})\\
&~~~~~ + \frac{1}{m_{(\ell)}}\sum^{m_{(1)}}_{j=1}\left\{\left(g_{\widehat{\bm{\beta}}_{(-\ell)}}(\widetilde{X}_{(j, \ell)}) - g_{\bm{\beta}}(\widetilde{X}_{(j, \ell)})\right)^2\right\},
\end{align*}
and $\widehat{\bm{\beta}}_{(-\ell)} := \argmin_{\bm{\beta} \in \Theta}\frac{1}{n_{(-\ell)}}\sum^{n_{(-\ell)}}_{i=1} \big(Y_{(i, \ell)} -g_{\bm{\beta}}\left(X_{(i, \ell)}\right)\big)^2$; that is, we replace $\widehat{f}_{(-\ell)}$ in the DR estimator with $g_{\widehat{\bm{\beta}}_{(-\ell)}}$.
We refer to this estimator as SDB estimator because it conducts DML by using a model of interest $g_{\widehat{\bm{\beta}}_{(-\ell)}}$ itself.

Here, we make the following assumption. 
\begin{assumption}
\label{asm:nuisance_conv2}
    For $\ell \in \{1,2\}$, $\left\|\widehat{r}_{(-\ell)} - r_0\right\|_2 = o_p(1)$ holds, 
as $n, m \to \infty$.
\end{assumption}

Then, we obtain the following corollary.
\begin{corollary}
\label{thm:main2}
Suppose that Assumptions~\ref{asm:bounded_output}--\ref{eq:unique_beta}, \ref{asm:linear_dep} and \ref{asm:nuisance_conv2} hold. If $m = \rho n$ for some universal constant $\rho$, then the SDB estimator $\widehat{\bm{\beta}}^{\mathrm{SDB}}$ has the following asymptotic distribution:
    \begin{align*}
        &\sqrt{n}\left(\widetilde{\bm{\beta}}_0 - \widehat{\bm{\beta}}^{\mathrm{SDB}}\right)\xrightarrow{\mathrm{d}} \mathcal{N}\left(\bm{0}, \widetilde{\Omega}\right),
    \end{align*}
    where
    \begin{align*}
        &\widetilde{\Omega} = \mathbb{E}_q\left[\dot{g}_{\widetilde{\bm{\beta}}_0}(\widetilde{X})\dot{g}^\top_{\widetilde{\bm{\beta}}_0}(\widetilde{X})\right]^{-1}\mathbb{E}_q\left[\sigma^2(\widetilde{X})r(\widetilde{X})\dot{g}_{\widetilde{\bm{\beta}}_0}(\widetilde{X})\dot{g}^\top_{\widetilde{\bm{\beta}}_0}(\widetilde{X})\right]\mathbb{E}_q\left[\dot{g}_{\widetilde{\bm{\beta}}_0}(\widetilde{X})\dot{g}^\top_{\widetilde{\bm{\beta}}_0}(\widetilde{X})\right]^{-1}.
    \end{align*}
    If $f_0 \in \mathcal{G}$, then $sqrt{n}\left(\bm{\beta}_0 - \widehat{\bm{\beta}}^{\mathrm{SDB}}\right)\xrightarrow{\mathrm{d}} \mathcal{N}\left(\bm{0}, \widetilde{\Omega}\right)$.
\end{corollary}
This estimator does not have the DR property; that is, we can adapt a covariate shift only when $r_0$ is consistently estimated. In this sense, the property is the same as the conventional covariate shift adaptation. However, unlike it, the SDB estimator reduces the bias caused by the density ratio estimator, and we can show the fast $\sqrt{n}$-convergence rate.

\section{Related Work}
\label{sec:related}
Covariate shift adaptation is widely studied in machine learning and statistics, primarily motivated by addressing model misspecification. A pivotal work is by \citet{Shimodaira2000}, as discussed in our Section~\ref{sec:asymptotic_property}. They advance covariate shift adaptation by introducing an importance weighting method for parametric models under model misspecification. Building upon this, \citet{sugiyama2008direct} proposes a direct density ratio estimation method for covariate shift adaptation, drawing from \citet{Shimodaira2000}'s formulation.

While \citet{Reddi2015} also contemplates DR covariate shift adaptation, their approach diverges from ours. Although they also refer to their method as the DR method, leading to potential naming overlaps, our estimator is traditionally dubbed the DR estimator within semiparametric inference literature. Thus, to maintain consistency, we also refer to our estimator as the DR estimator despite the risk of confusion.

We introduce applications of covariate shift adaptation. \citet{Sugiyama2006} and \citet{Sugiyama2009} leverage covariate shift adaptation combined with importance weighting to address active learning in the presence of model misspecification. \citet{Sugiyama2005} suggests a model selection technique under such conditions. Furthermore, \citet{Yamada2010} distinguishes speakers under a covariate shift, \citet{matsui2018variable} utilizes covariate shift adaptation for medical datasets, \citet{KatoUehara2020} executes off-policy evaluation, and \citet{Kato2020nnbr} classify text data. \citet{ramchandran2021ensembling} discusses appropriate models under a covariate shift. Batch-normalization can be justified as a covariate shift adaptation \citet{Sergey2015}. For the benchmark datasets, see \citet{koh2021wilds}.
For the inference, \citet{Tibshirani2019} proposed conformal inference under a covariate shift. 

Various approaches of direct density ratio estimation have been explored without undergoing density estimation \citep{Sugiyama2012}. These approaches include the moment matching \citep{Huang2007,Gretton2009}, the probabilistic classification \citep{Qin1998,Cheng2004}, the density matching \citep{Nguyen2010}, and the density-ratio fitting \citep{Kanamori2009}. \citet{Sugiyama2011bregman} shows that these methods could be generalized as density ratio matching under the Bregman divergence \citep{Bregman1967}. Recently, to estimate the density ratio with high-dimensional models, \citet{Rhodes2020} proposes the telescoping density-ratio estimation, and \citet{Kato2020nnbr} proposes using the non-negative Bregman divergence.

\section{Experiments}
\label{sec:exp}
This section provides simulation studies to verify the soundness of our proposed method. We compared our method (DR (CF)) with the OLS, weighted least squares for covariate shift adaptation with importance weighting and estimated density ratio by uLSIF \citep{Kanamori2009} (WLS), and nonparametric regression by the Kernel ridge regression (NP). 

For $X$, let $d = 3$, and its first element be the constant term. Let $\bm{\beta} = (\beta_0\ \beta_1\ \cdots \ \beta_k)$ be parameters, and $\bm{\beta}_0 = (\beta_{0, 0}\ \beta_{1, 0}\ \cdots \ \beta_{k, 0})$ be the true parameters.
For $x = (1\ x_1\ x_2)^\top$, we conduct simulation studies for the following two regression functions:
\begin{align*}
    &\mathrm{Model}~1:\ f_0(x) = \beta_{0, 0} + \beta_{1, 0}x_1 + \beta_{2, 0}x^2_1 + \beta_{3, 0}x_2\\
    &\ \ \ \ \ \ \ \ \ \ \ \ \ \ \ \ \ \ \ \ \ \ \ \ \ \ \ \ \ \ \ \ \ \ + \beta_{4, 0}x^2_{2, 0}+ 2\beta_{5, 0} x_1 x_2,\\
    &\mathrm{Model}~2:\ f_0(x) = \frac{1}{1 + \exp\left(\beta_{0, 0} + \beta_{1, 0}x_1 + \beta_{2, 0}x_2\right)}.
\end{align*}
In Model~1, we use $\beta_{0, 0} = \beta_{1, 0} = \cdots = \beta_{5, 0}$ and $\varepsilon^d$ follows the standard normal distribution. In Model~2, we use  $(\beta_{0, 0}, \beta_{1, 0}, \beta_{2, 0}) = (0, 2, 3)$, and $Y_i$ is generated from a Bernoulli distribution with probability $f_0(X_i)$. 

For each $f_0$ in the OLS, WLS, and DR, we train misspecified (Misspecified) and correctly specified models (Correct). For misspecified models, we use  
$
    g_{\bm{\beta}} = \beta_0 + \beta_1x_1 + \beta_2x_2.
$
For correctly specified models, we use the same linear models as $f_0$. 

For $X_i = (X_{1, i}\ X_{2, i})^\top$ and  $\widetilde{X}_j = (\widetilde{X}_{1, j}\ \widetilde{X}_{2, j})^\top$, we generate $(X_{1, i}\ X_{2, i})^\top$ and $(\widetilde{X}_{1, j}\ \widetilde{X}_{2, j})^\top$ from two-dimensional normal distributions $\mathcal{N}(\bm{\theta}, \Omega)$ and $\mathcal{N}(\widetilde{\bm{\theta}}, \Omega)$, respectively. We use two $\Omega$: $\Omega = \begin{pmatrix}
    1 & 0\\
    0 & 1
\end{pmatrix}$ and $\Omega = \begin{pmatrix}
    1 & 0.1\\
    0.1 & 1
\end{pmatrix}$, referred to as ``Indep $X$'' and ``Corr $X$'' in Table~\ref{tbl:exp_res}.  Here, $\bm{\theta}$ and $\widetilde{\bm{\theta}}$ are generated from a uniform distribution with a support $[-1, 1]$, respectively. We obtain $1,000$ observations and $500$ observations for train and test data. We report the mean squared errors (MSEs) and the standard deviation (parenthesis) in Table~\ref{tbl:exp_res}. More detailed results and additional experiments are shown in Appendix~\ref{appdx:exp}, including experiments with the SDB estimators.

Our DR outperforms the OLS and the WLS under misspecified models. However, under correctly specified models, the performance is lower than theirs. This deterioration is caused by additional errors by using cross-fitting. Compared to the OLS, WLS, and DML, the NP achieves the lower MSEs, as expected, although we focus on cases where the NP cannot be used (we need to use some restricted models for $\mathcal{G}$). This is because the NP uses a large class for the regression models. Combining Table~\ref{tbl:exp_res} with the results in Appendix~\ref{appdx:exp}, we consider that the improvement under misspecified models is mainly due to the existence of nonparametric models in the DR estimator. 

\section{Conclusion}
\label{sec:concl}
In this study, we proposed the DR covariate shift adaptation with importance weighting. Our method is robust to the inconsistency of the density ratio estimation and reduces bias incurred from the estimation error of the density ratio. We show the double robustness and the asymptotic distribution of our proposed estimator. Finally, we confirmed the soundness of our proposed method by simulation studies.

\bibliography{predictiontest}
\bibliographystyle{BoBWCSA.bbl}

\clearpage

\appendix

\onecolumn

\section{Related Work}
\label{appdx:related}

Our technique shares similarities with \citet{KatoUehara2020}, which employs DML in covariate shift adaptation for policy evaluation and learning. Specifically, we employ DML for covariate shift adaptation with a general parameterized regression function.
\citet{chernozhukov2023automatic} also applies DML for covariate shift adaptation, independently of our work. Although the idea of using DML for this task is similar to ours and \citet{KatoUehara2020}, their research diverges from them in that they focus on a linear regression model and automatic debiased learning by \citet{Chernozhukov2022}, without bypassing the density-ratio estimation. They also derive a different asymptotic distribution from ours, and it is an open issue what yields the difference and what implication each result has. We also address their open issue for applying the DML to classification by showing applicability to the logistic regression in Remark~\ref{rem:logit}. 

Importance weighting is widely used in statistical methods. \citet{Shimodaira2000} uses it for covariate shift adaptation. 
\citet{Cortes2010} also discusses the use of importance weighting under statistical learning. In causal inference, methods based on importance-weighting are known as inverse-probability weighting estimators \citet{Horowitz1994}, and \citet{KatoUehara2020} generalizes the causal inference methods by adapting a covariate shift.  

\citet{Kalan2020}, \citet{Zhang2022}, and \citet{Lei2021} investigate the optimality of the use of parametric models under a covariate shift.  
\citet{anonymous2023maximum} confirms the relationship between the problems of model specification and covariate shift adaptation. 

By using nonparametric models, we can adapt a covariate shift. This idea has been traditionally used in the context of model specification tests \citep{WOOLDRIDGE1990331}. 
\citet{Kpotufe2021} studies the use of nonparametric models. 
\citet{Ma2023} shows the minimax optimality of the kernel ridge regression under a covariate shift. \citet{Pathak2022} investigates a similarity measure for covariate shift adaptation using neural networks. \citet{schmidthieber2022local} proves a certain uniform optimality of nonparametric estimators based on neural networks. 

Although we focused on importance weighting, there are other approaches for covariate shift adaptation, including pseudo-labeling 
\citep{wang2023pseudolabeling} and node-based Bayesian neural networks \citep{Trinh2022}. \citet{BenDavid2006}, \citet{BenDavid2010}, \citet{Zhao2019OnLI}, and \citet{ruan2022optimal} investigates optimal representation for domain generalization and covariate shift adaptation.

\section{DR Covariate Shift Adaptation}
\label{appdx:pseudo}
This section provides a pseudo-code for our proposed DR covariate shift adaptation and a brief introduction of density-ratio estimation. 

\subsection{Pseudo-code for the DR Covariate Shift Adaptation}
While we introduced the DR covariate shift adaptation with $2$-fold cross-fitting in Section~\ref{sec:doublyrobust}, this section provides general $\xi$-fold cross-fitting for some $\xi\in\mathbb{N}$ such that $2 \leq \xi < \infty$. 

\begin{algorithm}[tb]
   \caption{\small{DR covariate shift adaptation}}
   \label{alg:dml}
\begin{algorithmic}
    \STATE \textbf{Input}: Train data $\{(Y_i, X_i)\}^{n}_{i=1}$ and test data $\{\widetilde{X}_j\}^{m}_{j=1}$. Parametric regression models $\mathcal{G}$. 
    \STATE Take a $\xi$-fold random partition  and $\left\{D_{(\ell)}\right\}^{\xi}_{\ell = 1}$ such that $\widetilde{D}_{(\ell)} = \{(Y_{(i, \ell)}, X_{(i, \ell)})\}^{n_{(\ell)}}_{i=1}$ and the size of each fold $\widetilde{D}_{(\ell)}$ is $n_{\xi} = n_{(\ell)}$.
    \STATE Take a $\xi$-fold random partition $\left\{\widetilde{D}_{(\ell)}\right\}^{\xi}_{\ell = 1}$ such that $\widetilde{D}_{(\ell)} = \{\widetilde{X}_{(j, \ell)}\}^{m_{(\ell)}}_{j=1}$ and the size of each fold $\widetilde{D}_{(\ell)}$ is $m_{\xi} = m_{(\ell)}$.
    \STATE For each $\ell\in[\xi]$, define $D_{(-\ell)} := \left\{D_{(\ell)}\right\}_{\ell' \neq \ell}$ and $\widetilde{D}_{(-\ell)} =  \left\{\widetilde{D}_{(-\ell)}\right\}_{\ell' \neq \ell}$.
    \FOR{$\ell\in[\xi]$}
    \STATE Construct estimators $\widehat{r}_{(-\ell)}$ and $\widehat{f}_{(-\ell)}$ using $D_{(-\ell)}$ and $\widetilde{D}_{(-\ell)}$.
    \STATE Construct a risk estimator 
    \begin{align*}
    \widehat{\mathcal{R}}_{(\ell)}(\bm{\beta}) &:= \frac{1}{n_{(2)}}\sum^{n_{(\ell)}}_{i=1}\Big\{\left(Y_{(i, \ell)} - g_{\bm{\beta}}(X_{(i, \ell)}) \right)^2 - \left(\widehat{f}_{(-\ell)}(X_{(i, \ell)}) - g_{\bm{\beta}}(X_{(i, \ell)})\right)^2\Big\} \widehat{r}_{(-\ell)}(X_{(i, \ell)})\\
    &\ \ \ \ \ \ \ \ \ + \frac{1}{m_{(\ell)}}\sum^{m_{(\ell)}}_{j=1}\left\{\left(\widehat{f}_{(-\ell)}(\widetilde{X}_{(j, \ell)}) - g_{\bm{\beta}}(\widetilde{X}_{(j, \ell)})\right)^2\right\}.
    \end{align*}
    \ENDFOR
    \STATE Construct an estimator 
    \begin{align*}
    \widehat{\bm{\beta}}^{\mathrm{DR}} = \argmin_{\bm{\beta} \in \Theta}\Big\{\widehat{\mathcal{R}}_{(2)}(\bm{\beta})\Big\}_{\ell \in [\xi]}.
\end{align*}
\end{algorithmic}
\end{algorithm}

\subsection{Density Ratio Estimation}
This section provides the formulation of the uLSIF, a method for density-ratio estimation, as a reference.

Let $\mathcal{S}$ be a class of non-negative measurable functions $s:\mathcal{X}\to \mathbb{R}^+$. We consider minimizing the following squared error between $s$ and $r$:
\begin{align}
\label{dr}
&\mathbb{E}_{p}[(s(X) - r_0(X_i))^{2}] = \mathbb{E}_{p}[r^2(X)] - 2\mathbb{E}_{q}[s(\widetilde{X})] + \mathbb{E}_{p}[s^2(X)].
\end{align}
The first term of the last equation does not affect the result of minimization, and we can ignore the term, i.e., the density ratio is estimated through the following minimization problem: 
\begin{align*}
s^{*} &= \argmin_{s\in\mathcal{S}} \left\{\mathbb{E}_{p}[r^2(X)] - 2\mathbb{E}_{q}[s(X)] + \mathbb{E}_{p}[s^2(X)]\right\} = \argmin_{s\in\mathcal{S}} \left\{\frac{1}{2} \mathbb{E}_{p}[s^2(X)] - \mathbb{E}_{q}[s(\widetilde{X})]\right\}.
\end{align*}
By replacing the population values with sample approximations, we can estimate $r_0$ by using observations with least-squares. 

\section{Proof of Theorem~\ref{thm:asymp_dist}}
\label{appendixthm:asymp_dist}
Let $[b]_a$ be an $a$-th element of a vector $b$. 
To prove Theorem~\ref{thm:asymp_dist}, we use the following lemma. The proof is shown in Appendix~\ref{appdx:proof_lemma}.
\begin{lemma}
\label{lem:nuisance}
Under the same conditions in Theorem~\ref{thm:asymp_dist}, the following holds:
\begin{align}
\label{eq:target2}
    \sqrt{n}\Bigg[&\sum_{\ell\in\{1, 2\}}\left\{\frac{1}{n_{(\ell)}}\sum^{n_{(\ell)}}_{i=1}\left(Y_{(i, \ell)} - \widehat{f}_{(-\ell)}(X_{(i, \ell)})\right) \widehat{r}_{(-\ell)}(X_{(i, \ell)})\dot{g}_{\bm{\beta}_0}(X_{(i, \ell)})\right.\nonumber \\
    &\left. + \frac{1}{m_{(\ell)}}\sum^{m_{(\ell)}}_{j=1}\left(\widehat{f}_{(-\ell)}(\widetilde{X}_{(j, \ell)}) - g_{\bm{\beta}_0}(\widetilde{X}_{(j, \ell)})\right)\dot{g}_{\bm{\beta}_0}(\widetilde{X}_{(j, \ell)})\right\}\nonumber\\
    &\ \ \ - \sum_{\ell\in\{1, 2\}}\left\{\frac{1}{n_{(\ell)}}\sum^{n_{(\ell)}}_{i=1}\left(Y_{(i, \ell)} - f_0(X_{(i, \ell)})\right) r_0(X_{(i, \ell)})\dot{g}_{\bm{\beta}_0}(X_{(i, \ell)})\right\}\Bigg]_a = o_p(1).
\end{align}
\end{lemma}

Then, we show Theorem~\ref{thm:asymp_dist}.
\begin{proof}[Proof of Theorem~\ref{thm:asymp_dist}]
For the empirical risk, the first order condition of the optimal regression model is given as
\begin{align*}
    &\sum_{\ell\in\{1, 2\}}\Bigg\{\frac{1}{n_{(\ell)}}\sum^{n_{(\ell)}}_{i=1}\left(Y_{(i, \ell)} - \widehat{f}_{(-\ell)}(X_{(i, \ell)})\right) \widehat{r}_{(-\ell)}(X_{(i, \ell)})\dot{g}_{\widehat{\bm{\beta}}^{\mathrm{DR}}}(X_{(i, \ell)})\\
    &\ \ \ \ \ \ \ \ \ \ \ \ \ \ \ \ \ \ \ \ \ \ \ \ \ \ \ \ \ \ \ \ \ \ \ \ \ \ \ \ \ + \frac{1}{m_{(\ell)}}\sum^{m_{(\ell)}}_{j=1}\left(\widehat{f}_{(-\ell)}(\widetilde{X}_{(j, \ell)}) - g_{\widehat{\bm{\beta}}^{\mathrm{DR}}}(\widetilde{X}_{(j, \ell)})\right)\dot{g}_{\widehat{\bm{\beta}}^{\mathrm{DR}}}(\widetilde{X}_{(j, \ell)})\Bigg\} = \bm{0}.
\end{align*}

From the Taylor expansion, it holds that
\begin{align*}
    &\sum_{\ell\in\{1, 2\}}\left\{\frac{1}{n_{(\ell)}}\sum^{n_{(\ell)}}_{i=1}\left(Y_{(i, \ell)} - \widehat{f}_{(-\ell)}(X_{(i, \ell)})\right) \widehat{r}_{(-\ell)}(X_{(i, \ell)})\dot{g}_{\widehat{\bm{\beta}}^{\mathrm{DR}}}(X_{(i, \ell)}) \right. \\
    &~~~~~~~~~~ \left. + \frac{1}{m_{(\ell)}}\sum^{m_{(\ell)}}_{j=1}\left(\widehat{f}_{(-\ell)}(\widetilde{X}_{(j, \ell)}) - g_{\widehat{\bm{\beta}}^{\mathrm{DR}}}(\widetilde{X}_{(j, \ell)})\right)\dot{g}_{\widehat{\bm{\beta}}^{\mathrm{DR}}}(\widetilde{X}_{(j, \ell)})\right\}\\
    =&\sum_{\ell\in\{1, 2\}}\left\{\frac{1}{n_{(\ell)}}\sum^{n_{(\ell)}}_{i=1}\left(Y_{(i, \ell)} - \widehat{f}_{(-\ell)}(X_i)\right) \widehat{r}_{(-\ell)}(X_i)\dot{g}_{\bm{\beta}_0}(X_{(i, \ell)}) \right. \\ 
    &~~~~~~~~~~ \left. + \frac{1}{m_{(\ell)}}\sum^{m_{(\ell)}}_{j=1}\left(\widehat{f}_{(-\ell)}(\widetilde{X}_{(j, \ell)}) - g_{{\bm{\beta}}_0}(\widetilde{X}_{(j, \ell)})\right)\dot{g}_{{\bm{\beta}}_0}(\widetilde{X}_{(j, \ell)})\right\}\\
    &\ \ \ \ \ \ + \left({\bm{\beta}}_0 - \widehat{\bm{\beta}}^{\mathrm{DR}}\right)\sum_{\ell\in\{1, 2\}}\Bigg\{ \frac{1}{n_{(\ell)}}\sum^{n_{(\ell)}}_{i=1}\left(Y_{(i, \ell)} - \widehat{f}_{(-\ell)}(X_{(i, \ell)})\right) \widehat{r}_{(-\ell)}(X_{(i, \ell)})\ddot{g}_{{\bm{\beta}}_0}(X_{(i, \ell)})\\
    &~~~~~~~~~~ - \frac{1}{m_{(\ell)}}\sum^{m_{(\ell)}}_{j=1}\left\{ \dot{g}_{{\bm{\beta}}_0}(\widetilde{X}_{(j, \ell)})\dot{g}_{{\bm{\beta}}_0}(\widetilde{X}_{(j, \ell)}) + \left(\widehat{f}_{(-\ell)}(\widetilde{X}_j) - g_{{\bm{\beta}}_0}(\widetilde{X}_{(j, \ell)})\right)\ddot{g}_{{\bm{\beta}}_0}(\widetilde{X}_{(j, \ell)}) \right\} \Bigg\}. 
\end{align*}
Then, we have
\begin{align*}
    &\sqrt{n}\left({\bm{\beta}}_0 - \widehat{\bm{\beta}}^{\mathrm{DR}}\right) = \sqrt{n}AB,
\end{align*}
where
\begin{align*}
    A &= - \Bigg(\sum_{\ell\in\{1, 2\}}\Bigg\{ \frac{1}{n_{(\ell)}}\sum^{n_{(\ell)}}_{i=1}\left(Y_{(i, \ell)} - \widehat{f}_{(-\ell)}(X_{(i, \ell)})\right) \widehat{r}_{(-\ell)}(X_{(i, \ell)})\ddot{g}_{{\bm{\beta}}_0}(X_{(i, \ell)})\\
    &~~~~~~~~~~ - \frac{1}{m_{(\ell)}}\sum^{m_{(\ell)}}_{j=1}\left\{ \dot{g}_{{\bm{\beta}}_0}(\widetilde{X}_j)\dot{g}_{{\bm{\beta}}_0}(\widetilde{X}_{(j, \ell)}) + \left(\widehat{f}_{(-\ell)}(\widetilde{X}_{(j, \ell)}) - g_{{\bm{\beta}}_0}(\widetilde{X}_{(j, \ell)})\right)\ddot{g}_{{\bm{\beta}}_0}(\widetilde{X}_{(j, \ell)}) \right\} \Bigg\}\Bigg)^{-1},\\
    B &= \sum_{\ell\in\{1, 2\}}\left\{\frac{1}{n_{(\ell)}}\sum^{n_{(\ell)}}_{i=1}\left(Y_{(i, \ell)} - \widehat{f}_{(-\ell)}(X_{(i, \ell)})\right) \widehat{r}_{(-\ell)}(X_{(i, \ell)})\dot{g}_{{\bm{\beta}}_0}(X_{(i, \ell)}) \right. \\ 
    &~~~~~~~~~~ \left. + \frac{1}{m_{(\ell)}}\sum^{m_{(\ell)}}_{j=1}\left(\widehat{f}_{(-\ell)}(\widetilde{X}_{(j, \ell)}) - g_{{\bm{\beta}}_0}(\widetilde{X}_{(j, \ell)})\right)\dot{g}_{{\bm{\beta}}_0}(\widetilde{X}_{(j, \ell)})\right\}.
\end{align*}

In the following parts, we show
\begin{align}
\label{eq:proof_target1}
    A\xrightarrow{\mathrm{p}} -\mathbb{E}_q\left[\dot{g}^\top_{{\bm{\beta}}_0}(\widetilde{X})\dot{g}_{{\bm{\beta}}_0}(\widetilde{X})\right]^{-1},
\end{align}
and
\begin{align}
\label{eq:proof_target2}
    \sqrt{n}B \xrightarrow{\mathrm{d}}\mathcal{N}\left(0, \mathbb{E}_q\left[\sigma^2(\widetilde{X})r_0(\widetilde{X})\Big\{\dot{g}_{\bm{\beta}_0}(\widetilde{X})\dot{g}^\top_{\bm{\beta}_0}(\widetilde{X})\Big\}\right]\right), 
\end{align}
separately.

\paragraph{Proof of \eqref{eq:proof_target1}.}
Because $\hat{f}_{(-\ell)}\xrightarrow{\mathrm{p}} f_0$ as $n_{(-\ell)} \to \infty$, it holds that
\[A^{-1}\xrightarrow{\mathrm{p}} -\mathbb{E}\left[\dot{g}^\top_{{\bm{\beta}}_0}(\widetilde{X})\dot{g}_{{\bm{\beta}}_0}(\widetilde{X})\right]\] 
as $m_{(1)} \to \infty$. Therefore, if $\mathbb{E}\left[\dot{g}^\top_{{\bm{\beta}}_0}(\widetilde{X})\dot{g}_{{\bm{\beta}}_0}(\widetilde{X})\right]^{-1}$ exists, then we have
\[A\xrightarrow{\mathrm{p}} -\mathbb{E}_q\left[\dot{g}^\top_{{\bm{\beta}}_0}(\widetilde{X})\dot{g}_{{\bm{\beta}}_0}(\widetilde{X})\right]^{-1}\] 
as $m_{(1)} \to \infty$.

\paragraph{Proof of \eqref{eq:proof_target2}.}
Next, we show \eqref{eq:proof_target2}. To show \eqref{eq:proof_target2}, we consider the following decomposition:
\begin{align*}
    &B = \sum_{\ell\in\{1, 2\}}\left\{\frac{1}{n_{(\ell)}}\sum^{n_{(\ell)}}_{i=1}\left(Y_{(i, \ell)} - \widehat{f}_{(-\ell)}(X_{(i, \ell)})\right) \widehat{r}_{(-\ell)}(X_{(i, \ell)})\dot{g}_{\bm{\beta}_0}(X_{(i, \ell)}) \right. \\
    & ~~~~~~~~~~\left. + \frac{1}{m_{(\ell)}}\sum^{m_{(\ell)}}_{j=1}\left(\widehat{f}_{(-\ell)}(\widetilde{X}_{(j, \ell)}) - g_{\bm{\beta}_0}(\widetilde{X}_{(j, \ell)})\right)\dot{g}_{\bm{\beta}_0}(\widetilde{X}_{(j, \ell)})\right\}\\
    &=\sum_{\ell\in\{1, 2\}}\left\{\frac{1}{n_{(\ell)}}\sum^{n_{(\ell)}}_{i=1}\left(Y_{(i, \ell)} - \widehat{f}_{(-\ell)}(X_{(i, \ell)})\right) \widehat{r}_{(-\ell)}(X_{(i, \ell)})\dot{g}_{\bm{\beta}_0}(X_{(i, \ell)}) \right. \\ 
    & ~~~~~~~~~~\left. + \frac{1}{m_{(\ell)}}\sum^{m_{(\ell)}}_{j=1}\left(\widehat{f}_{(-\ell)}(\widetilde{X}_{(j, \ell)}) - g_{\bm{\beta}_0}(\widetilde{X}_{(j, \ell)})\right)\dot{g}_{\bm{\beta}_0}(\widetilde{X}_{(j, \ell)})\right\}\\
    & ~~~~~ - \sum_{\ell\in\{1, 2\}}\left\{\frac{1}{n_{(\ell)}}\sum^{n_{(\ell)}}_{i=1}\left(Y_{(i, \ell)} - f_0(X_{(i, \ell)})\right) r_0(X_{(i, \ell)})\dot{g}_{\bm{\beta}_0}(X_{(i, \ell)}) \right\} \\ 
    & ~~~~~ + \sum_{\ell\in\{1, 2\}}\left\{\frac{1}{n_{(\ell)}}\sum^{n_{(\ell)}}_{i=1}\left(Y_{(i, \ell)} - f_0(X_{(i, \ell)})\right) r_0(X_{(i, \ell)})\dot{g}_{\bm{\beta}_0}(X_{(i, \ell)})\right\}.
\end{align*}

From Lemma~\ref{lem:nuisance}, it holds that
\begin{align}
\sqrt{n}\Bigg[&\sum_{\ell\in\{1, 2\}}\left\{\frac{1}{n_{(\ell)}}\sum^{n_{(\ell)}}_{i=1}\left(Y_{(i, \ell)} - \widehat{f}_{(-\ell)}(X_{(i, \ell)})\right) \widehat{r}_{(-\ell)}(X_{(i, \ell)})\dot{g}_{\bm{\beta}_0}(X_{(i, \ell)}) \right. \nonumber \\ 
& ~~~~~~~~~~ \left. + \frac{1}{m_{(\ell)}}\sum^{m_{(\ell)}}_{j=1}\left(\widehat{f}_{(-\ell)}(\widetilde{X}_{(j, \ell)}) - g_{\bm{\beta}_0}(\widetilde{X}_{(j, \ell)})\right)\dot{g}_{\bm{\beta}_0}(\widetilde{X}_{(j, \ell)})\right\}\nonumber\\
    &\ \ \ - \sum_{\ell\in\{1, 2\}}\left\{\frac{1}{n_{(\ell)}}\sum^{n_{(\ell)}}_{i=1}\left(Y_{(i, \ell)} - f_0(X_{(i, \ell)})\right) r_0(X_i)\dot{g}_{\bm{\beta}_0}(X_{(i, \ell)})\right\}\Bigg]_a\nonumber\\
&=o_p(1).
\end{align}
Therefore, we obtain
\begin{align*}
    &\sqrt{n}B = \sqrt{n}\sum_{\ell\in\{1, 2\}}\frac{1}{n_{(\ell)}}\sum^{n_{(\ell)}}_{i=1}\left(Y_{(i, \ell)} - f_0(X_{(i, \ell)})\right) r_0(X_{(i, \ell)})\dot{g}_{\bm{\beta}_0}(X_{(i, \ell)}) + o_p(1),
\end{align*}
which immediately yields \eqref{eq:proof_target2}.
\end{proof}

\section{Proof of Lemma~\ref{lem:nuisance}}
\label{appdx:proof_lemma}
To show \eqref{eq:target2}, we use the following decomposition:
\begin{align*}
    &\sqrt{n}\left\{\frac{1}{n_{(\ell)}}\sum^{n_{(\ell)}}_{i=1}\left(Y_{(i, \ell)} - \widehat{f}_{(-\ell)}(X_{(i, \ell)})\right) \widehat{r}_{(-\ell)}(X_{(i, \ell)})\dot{g}_{\bm{\beta}_0}(X_{(i, \ell)}) \right. \\ 
    & ~~~~~\left. + \frac{1}{m_{(\ell)}}\sum^{m_{(\ell)}}_{j=1}\left(\widehat{f}_{(-\ell)}(\widetilde{X}_{(j, \ell)}) - g_{\bm{\beta}_0}(\widetilde{X}_{(j, \ell)})\right)\dot{g}_{\bm{\beta}_0}(\widetilde{X}_{(j, \ell)})\right\}
    - \left\{\frac{1}{n_{(\ell)}}\sum^{n_{(\ell)}}_{i=1}\left(Y_{(i, \ell)} - f_0(X_{(i, \ell)})\right) r_0(X_i)\dot{g}_{\bm{\beta}_0}(X_{(i, \ell)})\right\}\\
    &= \sqrt{n}\left\{\frac{1}{n_{(\ell)}}\sum^{n_{(\ell)}}_{i=1}\left(Y_{(i, \ell)} - \widehat{f}_{(-\ell)}(X_{(i, \ell)})\right) \widehat{r}_{(-\ell)}(X_{(i, \ell)})\dot{g}_{\bm{\beta}_0}(X_{(i, \ell)}) + \frac{1}{m_{(\ell)}}\sum^{m_{(\ell)}}_{j=1}\widehat{f}_{(-\ell)}(\widetilde{X}_{(j, \ell)}\dot{g}_{\bm{\beta}_0}(\widetilde{X}_{(j, \ell)})\right\}\\
    & ~~~~~ - \left\{\frac{1}{n_{(\ell)}}\sum^{n_{(\ell)}}_{i=1}\left(Y_{(i, \ell)} - f_0(X_{(i, \ell)})\right) r_0(X_{(i, \ell)})\dot{g}_{\bm{\beta}_0}(X_{(i, \ell)}) + \frac{1}{m_{(\ell)}}\sum^{m_{(\ell)}}_{j=1}g_{\bm{\beta}_0}(\widetilde{X}_{(j, \ell)})\dot{g}_{\bm{\beta}_0}(\widetilde{X}_{(j, \ell)})\right\}\\
    &= \sqrt{n}\left\{\frac{1}{n_{(\ell)}}\sum^{n_{(\ell)}}_{i=1}\left(Y_{(i, \ell)} - \widehat{f}_{(-\ell)}(X_{(i, \ell)})\right) \widehat{r}_{(-\ell)}(X_{(i, \ell)})\dot{g}_{\bm{\beta}_0}(X_{(i, \ell)}) + \frac{1}{m_{(\ell)}}\sum^{m_{(\ell)}}_{j=1}\widehat{f}_{(-\ell)}(\widetilde{X}_{(j, \ell)})\dot{g}_{\bm{\beta}_0}(\widetilde{X}_{(j, \ell)})\right\}\\
    & ~~~~~ - \left\{\frac{1}{n_{(\ell)}}\sum^{n_{(\ell)}}_{i=1}\left(Y_{(i, \ell)} - f_0(X_{(i, \ell)})\right) r_0(X_{(i, \ell)})\dot{g}_{\bm{\beta}_0}(X_{(i, \ell)}) + \frac{1}{m_{(\ell)}}\sum^{m_{(\ell)}}_{j=1}g_{\bm{\beta}_0}(\widetilde{X}_{(j, \ell)})\dot{g}_{\bm{\beta}_0}(\widetilde{X}_{(j, \ell)})\right\}\\
    & ~~~~~ - \sqrt{n}\mathbb{E}\Bigg[\left\{\frac{1}{n_{(\ell)}}\sum^{n_{(\ell)}}_{i=1}\left(Y_{(i, \ell)} - \widehat{f}_{(-\ell)}(X_{(i, \ell)})\right) \widehat{r}_{(-\ell)}(X_{(i, \ell)})\dot{g}_{\bm{\beta}_0}(X_{(i, \ell)}) + \frac{1}{m_{(\ell)}}\sum^{m_{(\ell)}}_{j=1}\widehat{f}_{(-\ell)}(\widetilde{X}_{(j, \ell)})\dot{g}_{\bm{\beta}_0}(\widetilde{X}_{(j, \ell)})\right\}\\
    & ~~~~~~~~~~ - \left\{\frac{1}{n_{(\ell)}}\sum^{n_{(\ell)}}_{i=1}\left(Y_{(i, \ell)} - f_0(X_{(i, \ell)})\right) r_0(X_{(i, \ell)})\dot{g}_{\bm{\beta}_0}(X_{(i, \ell)}) + \frac{1}{m_{(\ell)}}\sum^{m_{(\ell)}}_{j=1}g_{\bm{\beta}_0}(\widetilde{X}_{(j, \ell)})\dot{g}_{\bm{\beta}_0}(\widetilde{X}_{(j, \ell)})\right\}\Bigg]\\
    &~~~~~ + \sqrt{n}\mathbb{E}\Bigg[\left\{\frac{1}{n_{(\ell)}}\sum^{n_{(\ell)}}_{i=1}\left(Y_{(i, \ell)} - \widehat{f}_{(-\ell)}(X_{(i, \ell)})\right) \widehat{r}_{(-\ell)}(X_{(i, \ell)})\dot{g}_{\bm{\beta}_0}(X_{(i, \ell)}) + \frac{1}{m_{(\ell)}}\sum^{m_{(\ell)}}_{j=1}\widehat{f}_{(-\ell)}(\widetilde{X}_{(j, \ell)})\dot{g}_{\bm{\beta}_0}(\widetilde{X}_{(j, \ell)})\right\}\\
    &~~~~~~~~~~ - \left\{\frac{1}{n_{(\ell)}}\sum^{n_{(\ell)}}_{i=1}\left(Y_{(i, \ell)} - f_0(X_{(i, \ell)})\right) r_0(X_{(i, \ell)})\dot{g}_{\bm{\beta}_0}(X_{(i, \ell)}) + \frac{1}{m_{(\ell)}}\sum^{m_{(\ell)}}_{j=1}g_{\bm{\beta}_0}(\widetilde{X}_{(j, \ell)})\dot{g}_{\bm{\beta}_0}(\widetilde{X}_{(j, \ell)})\right\}\Bigg].
\end{align*}

For a measurable function $\psi:\mathcal{X}\times \mathcal{R}\to\mathbb{R}$ and $\psi':\mathcal{X}\to\mathbb{R}$, let us define the empirical expectations over the labeled and unlabeled data as
\begin{align*}
&\mathbb{E}_{n_{(\ell)}}[\psi(X_{(i, \ell)}, Y_{(i, \ell)})] := \frac{1}{n_{(\ell)}}\sum^{n_{(\ell)}}_{i=1}\psi(X_{(i, \ell)}, Y_{(i, \ell)}),\qquad \mathbb{E}_{m_{(\ell)}}[\psi(\widetilde{X}_{(j, \ell)})] := \frac{1}{m_{(\ell)}}\sum^{m_{(\ell)}}_{j=1}\psi'(\widetilde{X}_{(j, \ell)}).
\end{align*} 

We denote 
\begin{align*}
&\phi_1(x, y; f, r) := (y-f(x))\dot{g}_{\bm{\beta}_0}(x)r(x),\qquad \phi_2(x; f) := \dot{g}_{\bm{\beta}_0}(x)f(x).
\end{align*}

Then, in the following proof, we separately show the following two inequalities:
\begin{align}
\label{eq:target4}
&\frac{\sqrt{n}}{\sqrt{n_{(1)}}}\mathbb{G}_{n_{(1)}}\big(\phi_1(X_{(i, 1)}, Y_{(i, 1)}; \widehat{f}_{(-1)}, \widehat{r}_{(-1)}) - \phi_1(X_{(i, 1)}, Y_{(i, 1)}; f_0, r_0)\big)\nonumber\\
&\ \ \ \ \ \ \ \ \ \ \ \ \ \ \ \ \ \ \ \ \ \ \ \ \ \ \ \ \ \ \ \ \ \ \ \ \ \ \ \ \ +\frac{\sqrt{n}}{\sqrt{m_{(1)}}}\mathbb{G}_{m_{(1)}}\big(\phi_2(\widetilde{X}_{(j, 1)}; \widehat{f}_{(-1)}) - \phi_2(\widetilde{X}_{(j, 1)}; f_0)\big) = \mathrm{o}_p(1),
\end{align}
and 
\begin{align}
\label{eq:target5}
&\sqrt{n}\Bigg\{\mathbb{E}_p\Big[\phi_1(X_{(i, 1)}, Y_{(i, 1)}; \widehat{f}_{(-1)}, \widehat{r}_{(-1)})\mid \widehat{f}_{(-1)}, \widehat{r}_{(-1)}\Big] - \mathbb{E}_p\Big[\phi_1(X_{(i, 1)}, Y_{(i, 1)}; f_0, r_0)\mid \widehat{f}_{(-1)}, \widehat{r}_{(-1)}\Big]\nonumber\\
& ~~~~~ + \mathbb{E}_q\Big[\phi_2(\widetilde{X}_{(j, 1)}; \widehat{f}_{(-1)})\mid \widehat{f}_{(-1)}\Big]- \mathbb{E}_q\Big[\phi_2(\widetilde{X}_{(j, 1)}; f_0)\mid \widehat{f}_{(-1)}\Big]\Bigg\} 
= \mathrm{o}_p(1),
\end{align}
where $\mathbb{G}_{n_{(\ell)}}$ and $\widetilde{\mathbb{G}}_{m_{(\ell)}}$ are empirical processes defined as 
\begin{align*}
&\mathbb{G}_{n_{(\ell)}}\big(\phi_1(X_{(i, \ell)}, Y_{(i, \ell)}; \widehat{f}_{(-\ell)}, \widehat{r}_{(-\ell)}) - \phi_1(X_{(i, \ell)}, Y_{(i, \ell)}; f_0, r_0)\big)\\
&\ \ \ =\sqrt{n_{(\ell)}}\big(\mathbb{E}_{n}[\phi_1(X_{(i, \ell)}, Y_{(i, \ell)}; \widehat{f}_{(-\ell)}, \widehat{r}_{(-\ell)})] - \mathbb{E}_p[\phi_1(X_{(i, \ell)}, Y_{(i, \ell)}; f_0, r_0)\mid f_0, r_0]\big)\\
&\mathbb{G}_{m_{(\ell)}}\big(\phi_2(\widetilde{X}_{(j, \ell)}; \widehat{f}_{(-\ell)}) - \phi_2(\widetilde{X}_{(j, \ell)}; f_0)\big)\\
&\ \ \ =\sqrt{m_{(\ell)}}\big(\mathbb{E}_{m}[\phi_2(\widetilde{X}_{(j, \ell)}; \widehat{f}_{(-\ell)})] - \mathbb{E}_q[\phi_2(\widetilde{X}_{(j, \ell)}; f_0)\mid f_0]\big).
\end{align*}

\paragraph{Step~1: Proof of \eqref{eq:target4}.}

If we can show that for any $\epsilon>0$, 
\begin{align}\
\label{eq:part}
 \lim_{n,m\to \infty}\mathbb{P}\Bigg[&\Big|\frac{\sqrt{n}}{\sqrt{n_{(1)}}}\mathbb{G}_{n_{(1)}}\big(\phi_1(X_{(i, 1)}, Y_{(i, 1)}; \widehat{f}_{(-1)}, \widehat{r}_{(-1)}) - \phi_1(X_{(i, 1)}, Y_{(i, 1)}; f_0, r_0)\big)\nonumber\\
&\ \ \ +\frac{\sqrt{n}}{\sqrt{m_{(1)}}}\mathbb{G}_{m_{(1)}}\big(\phi_2(\widetilde{X}_{(j, 1)}; \widehat{f}_{(-1)}) - \phi_2(\widetilde{X}_{(j, 1)}; f_0)\big) \Big| > \varepsilon \mid D_{(2)} \Bigg]=0,
\end{align}
then by the bounded convergence theorem, we would have 
\begin{align*}
\lim_{n \to \infty}P\Bigg[&\Big|\frac{\sqrt{n}}{\sqrt{n_{(1)}}}\mathbb{G}_{n_{(1)}}\big(\phi_1(X_{(i, 1)}, Y_{(i, 1)}; \widehat{f}_{(-1)}, \widehat{r}_{(-1)}) - \phi_1(X_{(i, 1)}, Y_{(i, 1)}; f_0, r_0)\big)\nonumber\\
&\ \ \ +\frac{\sqrt{n}}{\sqrt{m_{(1)}}}\mathbb{G}_{m_{(1)}}\big(\phi_2(\widetilde{X}_{(j, 1)}; \widehat{f}_{(-1)}) - \phi_2(\widetilde{X}_{(j, 1)}; f_0)\big) \Big| > \varepsilon \mid D_{(2)} \Bigg]=0,
\end{align*}
yielding the statement. 

To show \eqref{eq:part}, we show that the conditional mean is $0$ and the conditional variance is $\op(1)$. Then, \eqref{eq:part} is proved by the Chebyshev inequality following the proof of \citep[Theorem 4]{KallusUehara2019}. 
The conditional mean is 
\begin{align*}
&\mathbb{E}\Bigg[\frac{\sqrt{n}}{\sqrt{n_{(1)}}}\mathbb{G}_{n_{(1)}}\big(\phi_1(X_{(i, 1)}, Y_{(i, 1)}; \widehat{f}_{(-1)}, \widehat{r}_{(-1)}) - \phi_1(X_{(i, 1)}, Y_{(i, 1)}; f_0, r_0)\big)\\
&\ \ \ \ \ +\frac{\sqrt{n}}{\sqrt{m_{(1)}}}\mathbb{G}_{m_{(1)}}\big(\phi_2(\widetilde{X}_{(j, 1)}; \widehat{f}_{(-1)}) - \phi_2(\widetilde{X}_{(j, 1)}; f_0)\big) \mid D_{(2)}, \widetilde{D}_{(2)}\Bigg] \\
&= \mathbb{E}\Bigg[\frac{\sqrt{n}}{\sqrt{n_{(1)}}}\mathbb{G}_{n_{(1)}}\big(\phi_1(X_{(i, 1)}, Y_{(i, 1)}; \widehat{f}_{(-1)}, \widehat{r}_{(-1)}) - \phi_1(X_{(i, 1)}, Y_{(i, 1)}; f_0, r_0)\big),\\
&\ \ \ \ \ +\frac{\sqrt{n}}{\sqrt{m_{(1)}}}\mathbb{G}_{m_{(1)}}\big(\phi_2(\widetilde{X}_{(j, 1)}; \widehat{f}_{(-1)}) - \phi_2(\widetilde{X}_{(j, 1)}; f_0)\big) \mid \widehat{f}_{(-1)}, \widehat{r}_{(-1)}\Bigg] \\
&=0. 
\end{align*}

The conditional variance is bounded as 
\begin{align*}
&\mathrm{Var}\Bigg[\frac{\sqrt{n}}{\sqrt{n_{(1)}}}\mathbb{G}_{n_{(1)}}\big(\phi_1(X_{(i, 1)}, Y_{(i, 1)}; \widehat{f}_{(-1)}, \widehat{r}_{(-1)}) - \phi_1(X_{(i, 1)}, Y_{(i, 1)}; f_0, r_0)\big),\\
&\ \ \ \ \ +\frac{\sqrt{n}}{\sqrt{m_{(1)}}}\mathbb{G}_{m_{(1)}}\big(\phi_2(\widetilde{X}_{(j, 1)}; \widehat{f}_{(-1)}) - \phi_2(\widetilde{X}_{(j, 1)}; f_0)\big) \mid D_{(2)}, \widetilde{D}_{(2)}\Bigg]\\
&=\frac{n}{n_{(1)}}\mathrm{Var}\Bigg[\phi_1(X_{(i, 1)}, Y_{(i, 1)}; \widehat{f}_{(-1)}, \widehat{r}_{(-1)}) - \phi_1(X_{(i, 1)}, Y_{(i, 1)}; f_0, r_0)\mid D_{(2)}, \widetilde{D}_{(2)}\Bigg]\\
&\ \ \ \ \ +\frac{n}{m_{(1)}}\mathrm{Var}\Bigg[\phi_2(\widetilde{X}_{(j, 1)}; \widehat{f}_{(-1)}) - \phi_2(\widetilde{X}_{(j, 1)}; f_0)\mid \widehat{f}_{(-1)}, \widehat{r}_{(-1)} \mid D_{(2)}, \widetilde{D}_{(2)}\Bigg]\\
&=\frac{n}{n_{(1)}}\mathbb{E}_p\Bigg[\left\{\phi_1(X_{(i, 1)}, Y_{(i, 1)}; \widehat{f}_{(-1)}, \widehat{r}_{(-1)}) - \phi_1(X_{(i, 1)}, Y_{(i, 1)}; f_0, r_0)\right\}^2 \mid D_{(2)}, \widetilde{D}_{(2)}\Bigg]\\
&\ \ \ \ \ +\frac{n}{m_{(1)}}\mathbb{E}_q\Bigg[\left\{\phi_2(\widetilde{X}_{(j, 1)}; \widehat{f}_{(-1)}) - \phi_2(\widetilde{X}_{(j, 1)}; f_0)\mid \widehat{f}_{(-1)}, \widehat{r}_{(-1)}\right\}^2  \mid D_{(2)}, \widetilde{D}_{(2)}\Bigg]\\
&=\mathrm{o}_p(1)
\end{align*}

Here, we used
\begin{align}
\label{eq:first_e}
    \frac{n}{n_{(1)}}\mathbb{E}_p\Bigg[\left\{\phi_1(X_{(i, 1)}, Y_{(i, 1)}; \widehat{f}_{(-1)}, \widehat{r}_{(-1)}) - \phi_1(X_{(i, 1)}, Y_{(i, 1)}; f_0, r_0)\right\}^2 \mid D_{(2)}, \widetilde{D}_{(2)}\Bigg]=\mathrm{o}_p(1),
\end{align}
and 
\begin{align} 
\label{eq:second_e}
    \frac{n}{m_{(1)}}\mathbb{E}_q\Bigg[\left\{\phi_2(\widetilde{X}_{(j, 1)}; \widehat{f}_{(-1)}) - \phi_2(\widetilde{X}_{(j, 1)}; f_0)\mid \widehat{f}_{(-1)}, \widehat{r}_{(-1)}\right\}^2  \mid D_{(2)}, \widetilde{D}_{(2)}\Bigg]=\mathrm{o}_p(1). 
\end{align}
The first equation \eqref{eq:first_e} is proved by 
\begin{align*}
    &\mathbb{E}_p\Bigg[\Big\{\dot{g}_{\bm{\beta}_0}\big(X_{(i, 1)}\big)\big(Y_{(i, 1)}-\widehat{f}_{(-1)}\big(X_{(i, 1)}\big)\big)\widehat{r}_{(-1)}\big(X_{(i, 1)}\big) \\ 
    &~~~~~ - \dot{g}_{\bm{\beta}_0}\big(X_{(i, 1)}\big)(Y_{(i, 1)}-f_0\big(X_{(i, 1)}\big))r_0\big(X_{(i, 1)}\big) \Big\}^2 \mid D_{(2)}, \widetilde{D}_{(2)}\Bigg]\\
    &=\mathbb{E}_p\Bigg[\Big\{\dot{g}_{\bm{\beta}_0}\big(X_{(i, 1)}\big)\big(Y_{(i, 1)}-\widehat{f}_{(-1)}\big(X_{(i, 1)}\big)\big)\widehat{r}_{(-1)}\big(X_{(i, 1)}\big) - \dot{g}_{\bm{\beta}_0}\big(X_{(i, 1)}\big)(Y_{(i, 1)}-f_0\big(X_{(i, 1)}\big))\widehat{r}_{(-1)}\big(X_{(i, 1)}\big)\\
    &\ \ \ \ \  + \dot{g}_{\bm{\beta}_0}(x)(Y_{(i, 1)}-f_0\big(X_{(i, 1)}\big))\widehat{r}_{(-1)}\big(X_{(i, 1)}\big) - \dot{g}_{\bm{\beta}_0}(x)(Y_{(i, 1)}-f_0\big(X_{(i, 1)}\big))r_0\big(X_{(i, 1)}\big)\Big\}^2 \mid D_{(2)}, \widetilde{D}_{(2)}\Bigg]\\
    &\leq 2\mathbb{E}_p\Bigg[\Big\{\dot{g}_{\bm{\beta}_0}\big(X_{(i, 1)}\big)\big(Y_{(i, 1)}-\widehat{f}_{(-1)}\big(X_{(i, 1)}\big)\big)\widehat{r}_{(-1)}\big(X_{(i, 1)}\big) \\ 
    &~~~~~~~~~~  - \dot{g}_{\bm{\beta}_0}\big(X_{(i, 1)}\big)(Y_{(i, 1)}-f_0\big(X_{(i, 1)}\big))\widehat{r}_{(-1)}\big(X_{(i, 1)}\big)\Big\}^2\mid D_{(2)}, \widetilde{D}_{(2)}\Bigg]\\
    &~~~~~ + 2\mathbb{E}_p\Bigg[\Big\{\dot{g}_{\bm{\beta}_0}\big(X_{(i, 1)}\big)(Y_{(i, 1)}-f_0\big(X_{(i, 1)}\big))\widehat{r}_{(-1)}\big(X_{(i, 1)}\big) \\
    &~~~~~~~~~~ - \dot{g}_{\bm{\beta}_0}\big(X_{(i, 1)}\big)(Y_{(i, 1)}-f_0\big(X_{(i, 1)}\big))r_0\big(X_{(i, 1)}\big)\Big\}^2 \mid D_{(2)}, \widetilde{D}_{(2)}\Bigg]\\
    &\leq 2C \big\| f_0 - \widehat{f}_{(-1)}\big\|^2_2 + 2 \times 4 R^2_{\max}\big\|\widehat{r}_{(-1)} - r_0 \big\|^2_2
\end{align*}
Here, we have used a parallelogram law from the second line to the third line. We have use $0<\widehat r_0<C$ and $Y_{(i, 1)}, |\widehat{f}_{(-1)}|<R_{\max}$ according to the Assumptions \ref{asm:bounded_output} and \ref{asm:bounded_dens_ratio} and convergence rate conditions, from the third line to the fourth line.  The second equation \eqref{eq:second_e} is proved by Jensen's inequality.

\paragraph{Step~2: Proof of \eqref{eq:target5}.}

We have 
\begin{align*}
&\Bigg|\mathbb{E}_p\Big[\phi_1(X_{(i, 1)}, Y_{(i, 1)}; \widehat{f}_{(-1)}, \widehat{r}_{(-1)})\mid \widehat{f}_{(-1)}, \widehat{r}_{(-1)}\Big] - \mathbb{E}_p\Big[\phi_1(X_{(i, 1)}, Y_{(i, 1)}; f_0, r_0)\mid \widehat{f}_{(-1)}, \widehat{r}_{(-1)}\Big]\\
&\ \ \ \ \ + \mathbb{E}_q\Big[\phi_2(\widetilde{X}_{(j, 1)}; \widehat{f}_{(-1)})\mid \widehat{f}_{(-1)}\Big]- \mathbb{E}_q\Big[\phi_2(\widetilde{X}_{(j, 1)}; f_0)\mid \widehat{f}_{(-1)}, \widehat{r}_{(-1)}\Big]\Bigg|\\
&= \Bigg|\mathbb{E}_p\Big[\dot{g}_{\bm{\beta}_0}\big(X_{(i, 1)}\big)\big(Y_{(i, 1)}-\widehat{f}_{(-1)}\big(X_{(i, 1)}\big)\big)\widehat{r}_{(-1)}\big(X_{(i, 1)}\big)\mid \widehat{f}_{(-1)}, \widehat{r}_{(-1)}\Big] \\ 
&~~~~~ - \mathbb{E}_p\Big[\dot{g}_{\bm{\beta}_0}\big(X_{(i, 1)}\big)\big(Y_{(i, 1)}-f_0\big(X_{(i, 1)}\big)\big)r_0\big(X_{(i, 1)}\big)\Big]\\
&~~~~~ + \mathbb{E}_q\Big[\dot{g}_{\bm{\beta}_0}\big(X_{(i, 1)}\big)\widehat{f}_{(-1)}(\widetilde{X}_{(j, 1)})\mid \widehat{f}_{(-1)}\Big] - \mathbb{E}_q\Big[\dot{g}_{\bm{\beta}_0}(x)f_0(\widetilde{X}_{(j, 1)})\Big]\Bigg|\\
&= \Bigg|\mathbb{E}_p\Big[\dot{g}_{\bm{\beta}_0}(\widetilde{X}_{(j, 1)})\big(f_0\big(X_{(i, 1)}\big)-\widehat{f}_{(-1)}\big(X_{(i, 1)}\big)\big)\widehat{r}_{(-1)}\big(X_{(i, 1)}\big)\mid \widehat{f}_{(-1)}, \widehat{r}_{(-1)}\Big] - 0\\
&\ \ \ \ \ + \mathbb{E}_q\Big[\dot{g}_{\bm{\beta}_0}(\widetilde{X}_{(j, 1)})\widehat{f}_{(-1)}(\widetilde{X}_{(j, 1)})\mid \widehat{f}_{(-1)}\Big] - \mathbb{E}_q\Big[\dot{g}_{\bm{\beta}_0}(\widetilde{X}_{(j, 1)})f_0(\widetilde{X}_{(j, 1)})\mid \widehat{f}_{(-1)}, \widehat{r}_{(-1)}\Big]\Bigg|\\
&= \Bigg|\mathbb{E}_p\Big[\dot{g}_{\bm{\beta}_0}\big(X_{(i, 1)}\big)\big(f_0\big(X_{(i, 1)}\big)-\widehat{f}_{(-1)}\big(X_{(i, 1)}\big)\big)\big(\widehat{r}_{(-1)}\big(X_{(i, 1)}\big) - r_0\big(X_{(i, 1)}\big)\big)\Big]\Bigg|\\
&=\mathrm{o}_p(n^{-1/2}).
\end{align*}
All elements of the vector is bounded by a universal constant.
Here, we have used \Holder's inequality $ \|fg \|_1 \leq  \|f \|_2  \|g \|_2$.

\section{Semi-Covariate Shift Adaptation}
Our remaining problem is a lower bound for the asymptotic variances of estimators of $\widetilde{\bm{\beta}}^*$, which is referred to as the semiparametric efficiency bounds. To discuss semiparametric efficiency bounds, we restrict a class of estimators into regular estimators \citep{bickel98,VaartA.W.vander1998As}. An efficiency bound is defined for an estimand under some posited models of the DGP \citep{bickel98}. If this posited model is a parametric model, it is equal to the \Cramer-Rao lower bound. When this posited model is non or semiparametric model, we can still define a corresponding \Cramer-Rao lower bound. For regular estimators, semiparametric efficiency bounds are given as follows.
\begin{theorem}
\label{thm:efficiency}
Consider linear regression models $Y = f_0(X) + \varepsilon / \sqrt{r_0(X)}$. Denote the second moment of the limiting distribution of $\sqrt{n}\Big(\widehat{\bm{\beta}}-\bm{\beta}_0\Big)$ by $\mathrm{AMSE}\big[\widehat{\bm{\beta}}\big]$. Then, any regular estimators $\widehat{\bm{\beta}}$ satisfy
$\mathrm{AMSE}\big[\widehat{\bm{\beta}}\big]\geq V^*$, where 
\begin{align*}
     V^* = \sigma^2\mathbb{E}_q\left[\dot{g}_{\bm{\beta}_0}(\widetilde{X})\dot{g}^\top_{\bm{\beta}_0}(\widetilde{X})\right]^{-1}.
\end{align*}
\end{theorem}
Compared to the semiparametric efficiency bound, our proposed estimator has the larger asymptotic variance; that is, our estimator is not efficient. This property is curious because \citet{KatoUehara2020} shows semiparametric efficiency of estimators under covariate shift in off-policy evaluation. We conjecture that the above lower bound can be improved, which is an open issue. 

To reduce the variance, we propose the following estimator, referred to as an estimator of semi-covariate shift adaptation (SCSA): , referred to as an estimator of semi-covariate shift adaptation (SCSA): $\widehat{\bm{\beta}}^{\mathrm{SCSA}, \alpha} := \argmin_{\bm{\beta} \in \Theta}\sum_{\ell\in\{1, 2\}}\widehat{\mathcal{R}}^{\mathrm{SCSA}}_{(\ell)}(\bm{\beta})$, 
where
\begin{align*}
\widehat{\mathcal{R}}^{\mathrm{SCSA}, \alpha}_{(\ell)}(\bm{\beta}) 
&:= \frac{1}{n_{(\ell)}}\sum^{n_{(\ell)}}_{i=1}\Big\{\left(Y_{(i, \ell)} - g_{\bm{\beta}}(X_{(i, \ell)}) \right)^2- \left(\widehat{f}_{(-\ell)}(X_i) - g_{\bm{\beta}}(X_{(i, \ell)})\right)^2\Big\} \widehat{r}^\alpha_{(-\ell)}(X_{(i, \ell)}) \\ 
&~~~~~ + \frac{1}{m_{(\ell)}}\sum^{m_{(1)}}_{j=1}\left\{\left(\widehat{f}_{(-\ell)}(\widetilde{X}_{(j, \ell)}) - g_{\bm{\beta}}(\widetilde{X}_{(j, \ell)})\right)^2\right\}.
\end{align*}

\begin{corollary}
\label{thm:main3}
Suppose that Assumptions~\ref{asm:bounded_output}--\ref{eq:unique_beta}, \ref{asm:linear_dep}--\ref{asm:nuisance_conv} hold. If $m = \rho n$ for some universal constant $\rho$ and $f_0 \in \mathcal{G}$, then the SCA estimator $\widehat{\bm{\beta}}^{\mathrm{SCSA}, \alpha}$ has the following asymptotic distribution:
    \begin{align*}
        &\sqrt{n}\left(\widetilde{\bm{\beta}}_0 - \widehat{\bm{\beta}}^{\mathrm{SCSA}, \alpha}\right)\xrightarrow{\mathrm{d}} \mathcal{N}\left(\bm{0}, \Omega_\alpha\right),
    \end{align*}
where
\begin{align*}
    &\Omega_\alpha = \sigma^2\mathbb{E}_q\left[\dot{g}_{\bm{\beta}_0}(\widetilde{X})\dot{g}^\top_{\bm{\beta}_0}(\widetilde{X})\right]^{-1} \mathbb{E}_q\left[r^{2\alpha - 1}(\widetilde{X})\dot{g}_{\bm{\beta}_0}(\widetilde{X})\dot{g}^\top_{\bm{\beta}_0}(\widetilde{X})\right]\mathbb{E}_q\left[\dot{g}_{\bm{\beta}_0}(\widetilde{X})\dot{g}^\top_{\bm{\beta}_0}(\widetilde{X})\right]^{-1}.
\end{align*}
When $\alpha = 1/2$, $\Omega_\alpha = \sigma^2\mathbb{E}_q\left[\dot{g}_{\bm{\beta}_0}(\widetilde{X})\dot{g}^\top_{\bm{\beta}_0}(\widetilde{X})\right]^{-1}$ holds, which matches the efficiency bound. 
\end{corollary}

\section{Additional Experimental Results}
\label{appdx:exp}
This section provides additional experimental results. 

\subsection{The DR Estimator}
First, we show the estimation errors of $f_0$ of the OLS, WLS, and the DR in Figures~\ref{fig:exp1_kernel} and Figure~\ref{fig:exp2_kernel}. In Figure~\ref{fig:exp1_kernel}, we conduct experiments with misspecified models. In Figure~\ref{fig:exp2_kernel}, we conduct experiments with correctly specified models. Additionally, we show the results in boxplots of Figure~\ref{fig:exp3_boxplot}. 

In addition to the OLS, the WLS, the NP, and the DR, we conduct experiments with four additional methods. The first method is covariate shift adaptation via importance weighting with the ross-fitting. Unlike the DR, we do not use estimators of $f_0$ to reduce the bias of the estimation error of $r_0$. The second method is covariate shift adaptation via nonparametric regression. In this method, we first estimate $f_0$ by using the nonparametric kernel ridge regression. Then, we regress the estimator $\widehat{f}(X_i)$ on $X_i$. We also apply the cross-fitting for the method. The third method is the DR covariate shift adaptation without the cross-fitting. The fourth method is covariate shift adaptation via nonparametric regression without the cross-fitting (the second method without the cross-fitting). We refer to these methods as the WLS (CF), the CSA-NP (CF), the DR, and the CSA-NP, respectively. We show the results in Table~\ref{tbl:exp_res2}. 

From the results, we conclude that the reason why the DR outperforms the OLS and the WLS is because of the existence of the nonparametric estimators of $f_0$. In fact, nonparametric estimators outperform the DR. We also find that the benefits of cross-fitting are limited. This result does not deny the necessity of the cross-fitting because it guarantees the $\sqrt{n}$-consistency.

\begin{figure}[h]
    \centering
    \begin{subfigure}[b]{0.45\textwidth}
        \centering
        \includegraphics[width=60mm]{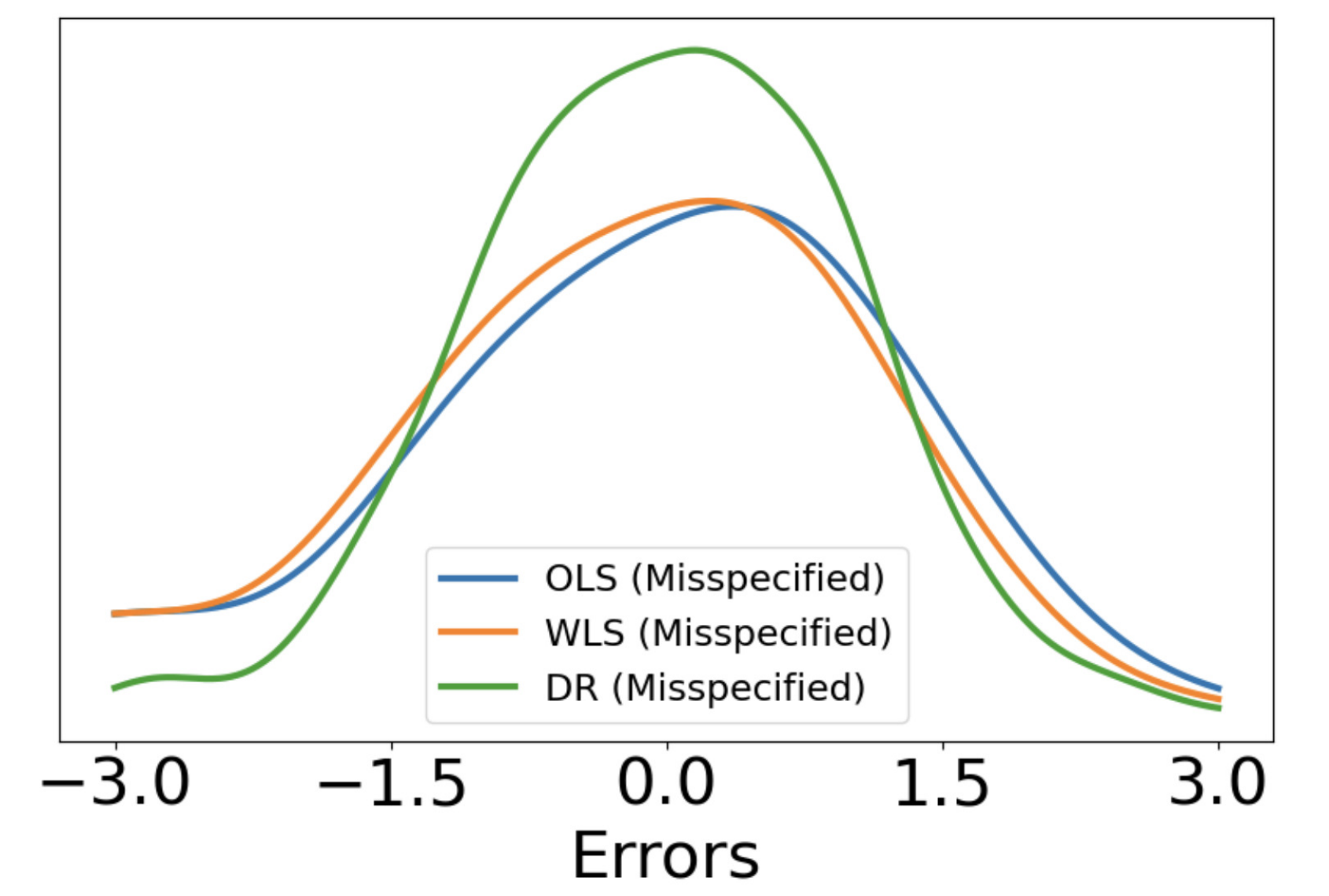}
        \caption{Estimation errors of simulation studies with misspecified models}
        \label{fig:exp1_kernel}
    \end{subfigure}
    \hfill
    \begin{subfigure}[b]{0.45\textwidth}
        \centering
       \includegraphics[width=60mm]{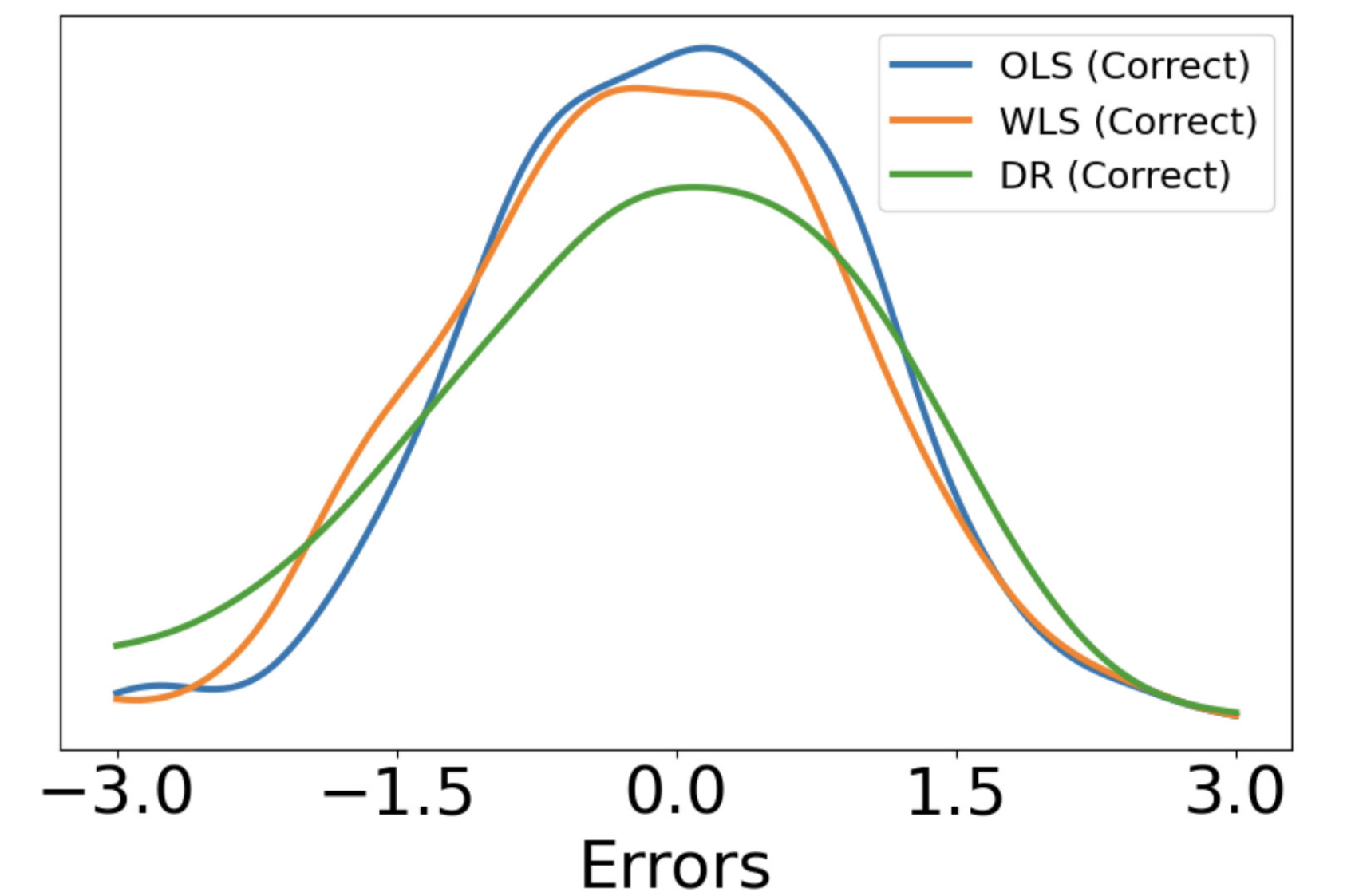}
        \caption{Estimation errors of simulation studies with correctly specified models}
        \label{fig:exp2_kernel}
    \end{subfigure}
=\end{figure}

\begin{figure*}[h]
  \centering
    \includegraphics[width=150mm]{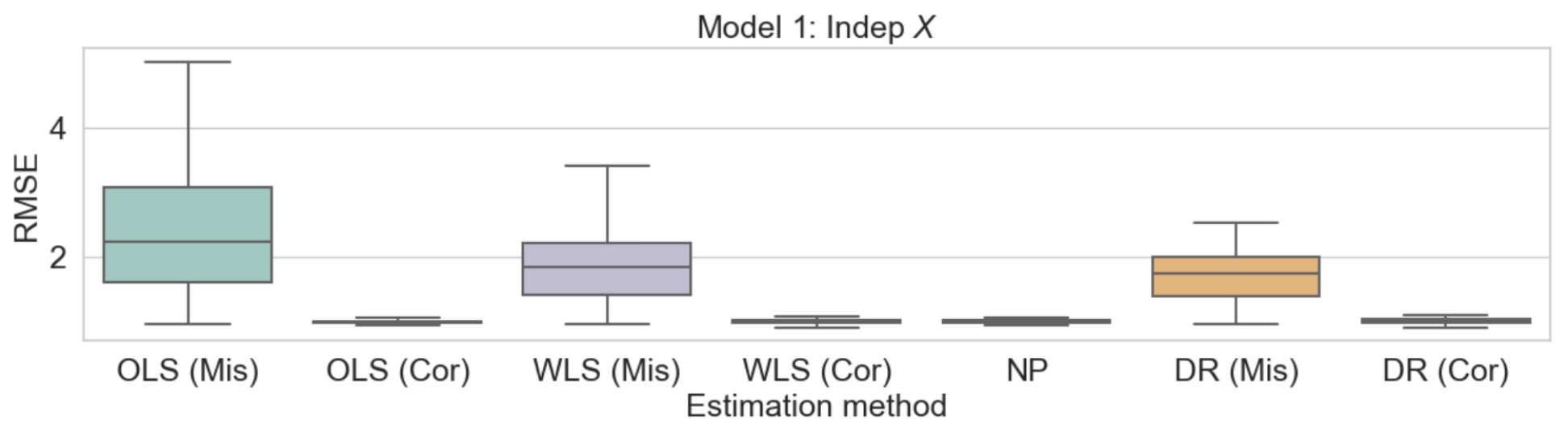}
\caption{RMSEs of experiments about the DR covariate shift adaptation.}
\label{fig:exp3_boxplot}
\end{figure*}

\begin{table*}[h!]
    \centering
    \caption{Results in Section~\ref{sec:exp}. The mean squared errors are reported.}
    \label{tbl:exp_res2}
    \scalebox{0.68}[0.68]{
    \begin{tabular}{|l|l|r|r|r|r|r|r|r|r|r|r|r|r|r|r|r|}
    \hline
       & &  \multicolumn{2}{|c|}{OLS}  &     \multicolumn{2}{|c|}{WLS}  &     \multicolumn{1}{|c|}{NP}  &       \multicolumn{2}{|c|}{DR (CF)}  &      \multicolumn{2}{|c|}{WLS (CF)}  &      \multicolumn{2}{|c|}{CSA-NP (CF)} &     \multicolumn{2}{|c|}{DR} &      \multicolumn{2}{|c|}{CSA-NP} \\
          \hline
   & & Mis & Cor & Mis & Cor & Mis & Mis & Cor & Mis & Cor & Mis & Cor & Mis & Cor  & Mis & Cor  \\
    \hline
    \multirow{2}{*}{Model~1} & Indep $X$ &  7.295 &  1.008 &  3.820 &  1.022 &  1.025 &  3.085 &  1.078 &  11.870 &  1.029 &  3.062 &  1.009 &  3.067 &  1.081 &  3.063 &  1.011 \\
    & Corr $X$ &  6.359 &  1.019 &  3.596 &  1.026 &  1.030 &  3.054 &  1.049 &   9.622 &  1.046 &  3.009 &  1.023 &  3.013 &  1.060 &  3.009 &  1.022 \\
    \multirow{2}{*}{Model~2} & Indep $X$ &  0.194 &  0.154 &  0.195 &  0.156 &  0.158 &  0.189 &  0.198 &   0.239 &  0.156 &  0.184 &  0.151 &  0.185 &  0.304 &  0.185 &  0.151 \\
    & Corr $X$ &  0.206 &  0.168 &  0.214 &  0.170 &  0.183 &  0.203 &  0.175 &   0.231 &  0.169 &  0.202 &  0.172 &  0.205 &  0.313 &  0.201 &  0.173 \\
    \bottomrule
    \end{tabular}
    }
\end{table*}

\subsection{The SDB Estimator}
Finally, we conduct experiments using the SDB estimator in Figure~\ref{fig:exp3_sdb}. Although the estimator has theoretically preferable properties, the performance is lower than the OLS and the WLS. This is because the preferable performance of the DR estimator stems from the existence of nonparametric estimators of $f_0$, and if we do not use them, the performance just worsens by using the SDB estimator. However, note that from the viewpoint of asymptotic theory, we can guarantee the performance for the SDB estimator by showing the $\sqrt{n}$-consistency. Therefore, as a tool for theoretical analysis, the SDB estimator is still useful.  

\begin{figure*}[h]
  \centering
    \includegraphics[width=140mm]{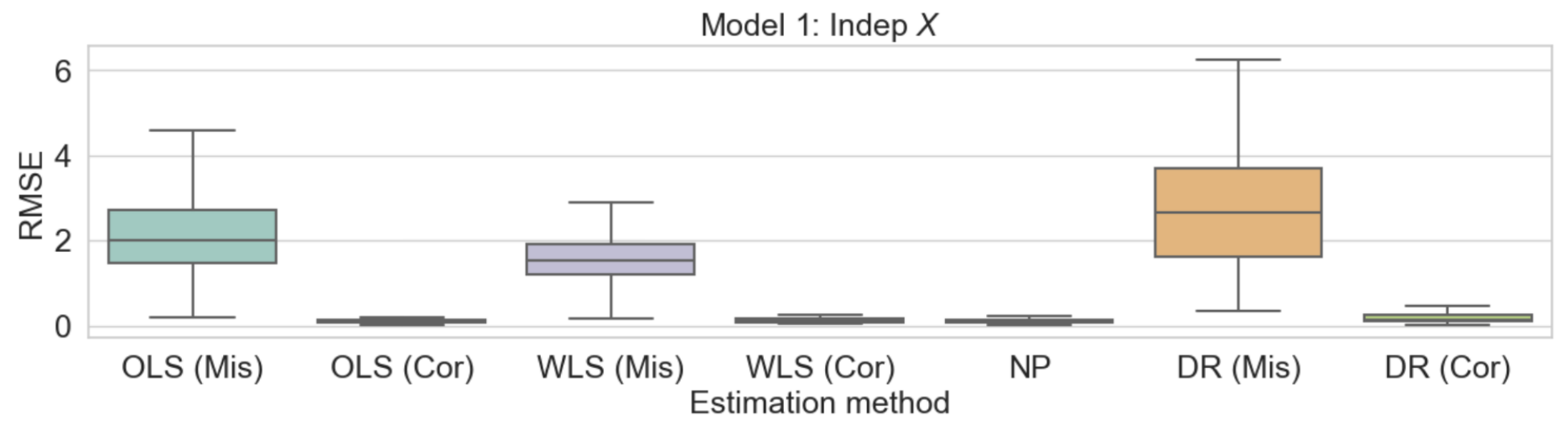}
\caption{RMSE of the SDB covariate shift adaptation.}
\label{fig:exp3_sdb}
\end{figure*}

\end{document}

%% file: math_commands.tex
\usepackage{amsmath,amsfonts,bm}

\def\eqref#1{equation~\ref{#1}}

\def\1{\bm{1}}

\DeclareMathAlphabet{\mathsfit}{\encodingdefault}{\sfdefault}{m}{sl}
\SetMathAlphabet{\mathsfit}{bold}{\encodingdefault}{\sfdefault}{bx}{n}

\DeclareMathOperator*{\argmin}{arg\,min}